\documentclass[journal]{IEEEtran}
\usepackage{cite}
\usepackage{tumcolor}
\usepackage{amsmath,amssymb,amsthm,fixmath}
\usepackage{mathtools}
\usepackage{algorithm}
\usepackage{algorithmic}
\usepackage{xcolor}
\usepackage{bm}
\usepackage{makecell}
\usepackage{pgfplots}
\usepackage{fancyhdr}

\usetikzlibrary{intersections,pgfplots.fillbetween,patterns,spy,shapes.geometric,calc,positioning}
\usetikzlibrary{arrows,decorations.pathmorphing,backgrounds,fit,matrix,fillbetween}
\usetikzlibrary{shapes.geometric, arrows.meta}
\usetikzlibrary{decorations.pathreplacing,decorations.markings}
\usepgfplotslibrary{groupplots}
\usepackage{subcaption}
\usepackage{tabularx}
\usepackage{soul}
\usepackage[acronym,shortcuts]{glossaries}
\usepackage{multirow}
\usepackage{multicol}
\usepackage{makecell}
\usepackage{hhline}
\usepackage{url}
\usepackage{float}
\usepackage{adjustbox}
\usepackage{graphicx}

\usetikzlibrary{external}
\newtheorem{theorem}{Theorem}
\newtheorem{lemma}{Lemma}

\definecolor{Purple}{rgb}{0.6350 0.0780 0.1840}
\definecolor{DarkGreen}{rgb}{0 0.50 0}
\definecolor{BrightGreen}{rgb}{0.71, 0.89, 0.58}

\newacronym{MMSE}{MMSE}{minimum mean squared error}
\newacronym{MSE}{MSE}{mean squared error}
\newacronym{NMSE}{NMSE}{normalized MSE}
\newacronym{LMMSE}{LMMSE}{linear minimum mean squared error}
\newacronym{WSS}{WSS}{wide sense stationary}
\newacronym{DFT}{DFT}{discrete Fourier transform}
\newacronym{EM}{EM}{expectation-maximization}
\newacronym{BIC}{BIC}{Bayesian information criterion}
\newacronym{AIC}{AIC}{Akaike information criterion}
\newacronym{i.i.d.}{i.i.d.}{independent and identically distributed}
\newacronym{FLOP}{FLOP}{floating point operation}
\newacronym{AWGN}{AWGN}{additive white Gaussian noise}
\newacronym{LS}{LS}{least squares}
\newacronym{SCM}{SCM}{sample covariance matrix}
\newacronym[firstplural = {covariance matrices (CMs)}]{CM}{CM}{covariance matrix}
\newacronym{ICM}{ICM}{inverse covariance matrix}
\newacronym{AR}{AR}{autoregressive}
\newacronym{MA}{MA}{moving-average}
\newacronym{ARMA}{ARMA}{autoregressive moving-average}
\newacronym{FBM}{FBM}{fractional Brownian motion}
\newacronym{FFT}{FFT}{fast Fourier transform}
\newacronym{GMM}{GMM}{gaussian mixture model}
 \newacronym{VAE}{VAE}{variational autoencoder}
 \newacronym{NN}{NN}{neural network}
 \newacronym{PD}{PD}{positive definite}
 \newacronym{OP}{OP}{optimization problem}
\newacronym{GS}{GS}{Gohberg-Semencul}
\newacronym{SAVG}{SAVG}{SCM averaged along its diagonals}
\newacronym{PLStext}{PLS}{projected LS}

\newacronym{Eig}{\texttt{Eig}}{eigenvalue}
\newacronym{Frob}{\texttt{Frob}}{Frobenius}
\newacronym{PGD}{\texttt{PGD}}{projected gradient descent}
\newacronym{PLS}{\texttt{PLS}}{projected LS}
\newacronym{band}{\texttt{Band}}{banding}
\newacronym{tape}{\texttt{Tape}}{tapering}
\newacronym{circ}{\texttt{Circ}}{}
\newacronym{em}{\texttt{EM}}{}
\newacronym{avg}{\texttt{Avg}}{}
\newacronym{shu}{\texttt{ShU}}{unbiased shrinkage estimator}
\newacronym{shb}{\texttt{ShB}}{biased shrinkage estimator}

\DeclareMathOperator*{\argmax}{\mathrm{argmax}}
\DeclareMathOperator*{\argmin}{\mathrm{argmin}}
\DeclareMathOperator*{\tr}{\mathrm{tr}}
\DeclareMathOperator*{\E}{\mathbb{E}}

\fancypagestyle{cfooter}{ %
\fancyhf{} % remove everything
\cfoot{\scriptsize{This work has been submitted to the IEEE for possible publication. Copyright may be transferred without notice, after which this version may no longer be accessible.}
}
}

\def\BibTeX{{\rm B\kern-.05em{\sc i\kern-.025em b}\kern-.08em
    T\kern-.1667em\lower.7ex\hbox{E}\kern-.125emX}}
\usepackage{balance}
\begin{document}

\title{Gohberg-Semencul Estimation of Toeplitz Structured Covariance Matrices and Their Inverses}

\author{Benedikt B\"ock, \IEEEmembership{Graduate Student Member, IEEE}, Dominik Semmler, \IEEEmembership{Graduate Student Member, IEEE}, \\ Benedikt Fesl, \IEEEmembership{Graduate Student Member, IEEE}, Michael Baur, \IEEEmembership{Graduate Student Member, IEEE}, \\ and Wolfgang Utschick, \IEEEmembership{Fellow, IEEE}
	\thanks{
		The authors are with Lehrstuhl f\"ur Methoden der Signalverarbeitung, Technische Universit\"at M\"unchen, 80333 M\"unchen, Germany  (e-mail: benedikt.boeck@tum.de; dominik.semmler@tum.de; benedikt.fesl@tum.de; mi.baur@tum.de; utschick@tum.de).
		
		%For additional proofs, the reader is referred to the extended version available in \cite{boeck2023}
	}

}

\maketitle

\thispagestyle{cfooter}

\begin{abstract}
When only few data samples are accessible, utilizing structural prior knowledge is essential for estimating covariance matrices and their inverses. One prominent example is knowing the covariance matrix to be Toeplitz structured, which occurs when dealing with \ac{WSS} processes. This work introduces a novel class of positive definiteness ensuring likelihood-based estimators for Toeplitz structured \acp{CM} and their inverses. In order to accomplish this, we derive positive definiteness enforcing constraint sets for the \ac{GS} parameterization of inverse symmetric Toeplitz matrices. Motivated by the relationship between the \ac{GS} parameterization and \ac{AR} processes, we propose hyperparameter tuning techniques, which enable our estimators to combine advantages from state-of-the-art likelihood and non-parametric estimators. Moreover, we present a computationally cheap closed-form estimator, which is derived by maximizing an approximate likelihood. Due to the ensured positive definiteness, our estimators perform well for both the estimation of the \ac{CM} and the \ac{ICM}. Extensive simulation results validate the proposed estimators' efficacy for several standard Toeplitz structured \acp{CM} commonly employed in a wide range of applications.

\end{abstract}


\section{Introduction}
Second moments in the form of the \ac{CM} and its inverse characterize the statistical properties of random variables and, thus, play a fundamental role in statistical signal processing. They are used in a wide range of applications to extract the desired information from collected data, including linear estimation, prediction, dimensionality reduction, and clustering \cite[Ch. 1-2]{Jolliffe1986},\cite[Ch. 3-4]{Kailath2000},\cite{Bester2011}. In practical settings, however, the required second moments are usually unknown and must be estimated \cite[Ch. 3]{Jolliffe1986},\cite{Bester2011}. 
A basic unbiased estimator of the \ac{CM} is the \ac{SCM}, which yields good performance in cases where many data samples are accessible. Nevertheless, its slowly decaying \ac{MSE} makes it unreliable in scenarios where the number of available data samples is small \cite[Sec. 2.2]{Pourahmadi2013}. Furthermore, in cases where the number of data samples is even smaller than their dimension, the \ac{SCM} becomes singular and, thus, cannot be directly used for estimating the \ac{ICM}. Prominent examples of applications in which generally only a very limited amount of samples is available are, e.g., financial engineering \cite{Markowitz1952}, array signal processing \cite{Guerci1999}, biological inference, and social networks \cite{Fan2013},\cite{Jianqing2014}. Thus, many different estimators for the \ac{CM} and \ac{ICM} based on few samples have been developed over the last decades. 

In this paper, we consider the problem of estimating Toeplitz structured \acp{CM} and their inverses, constituting a large subclass of second moments. Toeplitz structures arise when a \ac{WSS} random process is sampled equidistantly. This occurs in, e.g., sensor arrays capturing samples of a \ac{WSS} spatial process or sampling a \ac{WSS} time-varying signal. Many problems in imaging \cite{Snyder1989}, array signal processing \cite{Fuhrmann1991}, speech and audio processing \cite{Ephraim1989}, geostatistics \cite{Furrer2007}, and medical applications \cite{Derado2010} fit into this category. Moreover, many non-stationary processes in, e.g., wireless communication \cite[Sec. 2.3]{Stüber2017}, speech processing \cite{Paliwal2010}, and finance \cite{Heckens2020} , can be assumed to be quasi-stationary and, thus, also exhibit a locally concentrated Toeplitz structure within their second moments.   

Exploiting structural knowledge for estimating \acp{CM} was first considered in \cite{Anderson1973} and \cite{Burg1982}, in which conditions for gradient-based updates of the structured \ac{CM} are derived to maximize the Gaussian likelihood. The main focus in these publications lies on \acp{CM} that can be linearly decomposed by a dictionary of pre-known basis matrices (e.g., Toeplitz). Another likelihood estimator, proposed in \cite{Miller1987}, considers the available samples as incomplete observations from  data samples with larger dimension and a corresponding circulant \ac{CM}. Enforcing the \ac{CM} to be circulant exhibits a closed-form solution for the likelihood's maximizer and allows a computationally efficient way to estimate the Toeplitz structured \ac{CM} by appling the \ac{EM} algorithm for incomplete data \cite{Dempster1977}.

A different approach to estimate Toeplitz structured \acp{CM} is to exploit their property to be constant along the diagonals and to take the \ac{SAVG} as \ac{CM} estimate \cite{Cai2013}. Compared to likelihood estimators, this idea has the advantage that it neither assumes any underlying distribution nor requires solving an \ac{OP}. On the other hand, averaging along the diagonals does not necessarily preserve the \ac{CM}'s property of positive semidefiniteness. Additionally, as the distance from the off-diagonal to the main-diagonal increases, fewer samples are considered for averaging, resulting in a higher variance error.
%Moreover, the further out the diagonals are averaged, the fewer samples are taken into account resulting in a higher variance error of the entries. 
For that reason, the regularization techniques of banding and tapering the estimator have been analyzed \cite{Cai2013}. Banding describes the method of setting all matrix entries to zero in all off-diagonals above a certain bound. Tapering generalizes banding and multiplies the entries along the off-diagonals with decaying weights rather than setting a hard threshold. These techniques were first proposed for the \ac{SCM}, introducing a bias and leveraging the bias-variance trade-off  \cite{Bickel2008B},\cite{Cai2010}. A further advantage of these non-parametric estimators is that they perform well in the limit of just one sample \cite{Wu2009},\cite{Murry2010}, which is typically the setup in time series analysis \cite{Brockwell2016}. 
Another way of regularization is to shrink the \ac{SCM} towards the \ac{SAVG}, i.e., to build an optimized convex combination of both \cite{Lancewicki2014}. This idea was first proposed in \cite{Ledoit2004} for cases without prior structural information and improved in \cite{ChenWiesel2010} for Gaussian distributions. Generally, the adjustment of the convex weighting is based on a minimization of the estimated \ac{MSE}. 

A closely related topic to estimating Toeplitz structured \acp{CM} is parameter estimation of \ac{WSS} \ac{AR} processes in case of observing one sample \cite[Ch. 6]{Kay1988}. By applying the \ac{GS} decomposition, estimating the parameters of an \ac{AR} process implicitly yields an estimate of the \ac{ICM} (cf. Section \ref{sec:hyperparameter_section}). However, due to the difficulty to optimize the exact likelihood of \ac{AR} processes, the standard approach is to optimize the sample's conditional (approximate) likelihood by assuming initially observed entries to be deterministic, leading to a \ac{LS} problem \cite[Ch. 5]{Hamilton1994}. The resulting \ac{ICM} estimator, however, cannot be guaranteed to be \ac{PD} and, thus, inverting can yield arbitrarily bad \ac{CM} estimates. 

This work introduces a novel class of positive definiteness ensuring likelihood estimators for Toeplitz structured \acp{CM} and their inverses by utilizing the \ac{GS} parameterization for inverse symmetric Toeplitz matrices \cite{Mukherjee1988} and, thus, by effectively fitting an \ac{AR} process to the observed data. 
The closest work to ours is \cite{McWhorter1995}, which also yields an exact likelihood estimator based on the \ac{GS} decomposition. However, this method cannot efficiently distinguish solutions yielding positive definite and indefinite \acp{CM} and is either intractable for higher order \ac{AR} processes or only converges if the initialization is ``sufficiently'' close to the exact solution. 
In summary, our main contributions are the following:
\begin{itemize}
    \item We derive multiple positive definiteness ensuring constraint sets for the \ac{GS} parameterization. These sets include  simple box constraints, offering an efficient means of formulating computationally cheap likelihood-based estimators and integrating additional prior knowledge.
    \item We introduce hyperparameter tuning techniques enabling our estimators to share benefits from both state-of-the-art likelihood and non-parametric estimators.
    \item We establish an exact likelihood estimator, which provably converges to a local optimum and whose complexity only scales quadratically with the samples' dimension.
    \item We propose a computationally cheap closed-form estimator based on an approximate likelihood.
\end{itemize}

\textit{Notation:}
We begin indexing matrices and vectors with zero unless explicitly specified otherwise. We utilize a special case of the generalized Fibonacci sequence from \cite{Sahin2018} denoted as
\begin{equation}
\label{eq:fibonacci_definition}
    \begin{aligned}
        F_0(\bm{r}) =&\ 1 \\
        F_i(\bm{r}) =& \sum_{l=0}^{i-1} r_{i-l}F_l(\bm{r})\ \mathrm{for}\ \mathrm{all}\ i \in \{1,\ldots,P\}.
    \end{aligned}
\end{equation}
In this particular scenario, the vector $\bm{r} = [r_1,\ldots,r_P]^{\operatorname{T}}$ begins with the index $1$. The set of strictly positive real valued numbers, the set of strictly positive natural numbers and the convex cone of Hermitian $P \times P$ \ac{PD} matrices are denoted by $\mathbb{R}_{++}$, $\mathbb{N}$ and $\mathbb{S}^{P}_{++}$, respectively. The complex conjugate of $\beta \in \mathbb{C}$ is given by $\overline{\beta}$ and the operator $\operatorname{\mathbf{E}}$ denotes the shift-down matrix, defined by its entries $\operatorname{E}_{ij} = \delta(i-1-j)$ with $\delta(\cdot)$ being the Kronecker delta. We write the indicator function as $\chi(\cdot)$, which is one if the argument is true and zero otherwise. Moreover, by $\left|\bm{r}\right|$ we indicate the vector, which contains the entries of $\bm{r}$ in absolute value and the $i$,$j$th entry in some matrix $\bm{A}$ is denoted by $\bm{A}_{ij}$. Additionally, by writing $\bm{A}_{\geq i, \geq j}$ or $\bm{r}_{\geq k}$, we refer to the submatrix or subvector of $\bm{A}$ or $\bm{r}$, starting at row index $i$ and column index $j$, or index $k$, respectively.

\section{Problem Formulation}
\label{sec:problem_formulation}

Let $\mathcal{D} = \{\bm{x}^{(n)}\}_{n=1}^{N}$ be a set of $N$ \ac{i.i.d.} real or complex-valued samples drawn from a $P$-dimensional mean-zero Gaussian distribution with a full-rank Toeplitz structured \ac{CM} $\bm{C}$. This work aims to find an estimator $\hat{\bm{\Gamma}}(\mathcal{D})$ for the \ac{ICM} $\bm{\Gamma} = \bm{C}^{-1}$, given the dataset $\mathcal{D}$ and assuming $N$ either smaller or in the range of $P$.
Likelihood estimators maximize the probability density of the collected data over a parameter space containing all allowed configurations of the estimator. For our given setup, the likelihood estimator equals
\begin{equation}
    \label{eq:original_likelihood}
    \hat{\bm{\Gamma}}(\mathcal{D}) = \argmax\limits_{\bm{\Gamma}} \prod_{n=1}^N \mathcal{N}(\bm{x}_n;\bm{0},\bm{\Gamma}^{-1})\ \ \mathrm{s.t.}\ \bm{\Gamma} \in \mathcal{G}
\end{equation}
with $\mathcal{G}$ containing all inverses of \ac{PD} Hermitian Toeplitz matrices. After a few reformulations, by neglecting constants and introducing the \ac{SCM} $\bm{S} = 1/N \sum_{n=1}^N \bm{x}_n\bm{x}_n^{\operatorname{H}}$, the \ac{OP} in \eqref{eq:original_likelihood} can be stated as 
\begin{equation}
    \label{eq:compact_likelihood}
    \hat{\bm{\Gamma}}(\mathcal{D}) = \argmax\limits_{\bm{\Gamma}}\ \log \det \bm{\Gamma} - \tr(\bm{\Gamma} \bm{S})\ \ \mathrm{s.t.}\ \bm{\Gamma} \in \mathcal{G}.
\end{equation}
We specify the constraint set $\mathcal{G}$ by  using the \ac{GS} decomposition for symmetric Toeplitz matrices \cite{Mukherjee1988}. After a minor extension to the complex-valued case, it states that a matrix $\bm{T}^{-1}$ is the inverse of a full-rank but not necessarily \ac{PD} Hermitian Toeplitz matrix $\bm{T} \in \mathbb{C}^{P \times P}$ if and only if it exhibits a decomposition in the following form
\begin{equation}
    \label{eq:GS_formula}
    \bm{T}^{-1} = \frac{1}{\alpha_0}(\bm{B}\bm{B}^{\operatorname{H}} - \bm{Z}\bm{Z}^{\operatorname{H}}),
\end{equation}
where $\bm{B}$ and $\bm{Z}$ are defined as
\begin{equation}
	\label{eq:B_def}
	\bm{B} = \begin{pmatrix}
		\alpha_0 & 0 & 0 & \ldots & 0 \\
		\alpha_1 & \alpha_0 & 0 & \ldots & 0 \\
		\ldots & \ldots & \ldots & \ldots & \ldots \\
	\alpha_{P-1} & \alpha_{P-2} & \alpha_{P-3} & \ldots & \alpha_0
	\end{pmatrix}
\end{equation}
and 
\begin{equation}
	\label{eq:Z_def}
	\bm{Z} = \begin{pmatrix}
		0 & 0 & 0 & \ldots & 0 \\
		\overline{\alpha_{P-1}} & 0 & 0 & \ldots & 0 \\
		\ldots & \ldots & \ldots & \ldots & \ldots \\
		\overline{\alpha_{1}} & \overline{\alpha_{2}} & \ldots & \overline{\alpha_{P-1}} & 0
	\end{pmatrix}.
\end{equation}
The parameter $\alpha_0$ has to be real-valued and positive while the residual parameters $\alpha_1,\ldots,\alpha_{P-1}$ are unconstrained. We reparameterize the \ac{OP} in \eqref{eq:compact_likelihood} by enforcing $\bm{\Gamma}$ to satisfy the \ac{GS} formula in \eqref{eq:GS_formula}. By defining the corresponding log-likelihood $\mathcal{L}_{\mathcal{D}}(\bm{\alpha})$ as 
\begin{equation}
    \label{eq:loglikelihood_def}
    \mathcal{L}_{\mathcal{D}}(\bm{\alpha}) = \log \det \bm{\Gamma}(\bm{\alpha}) - \tr(\bm{\Gamma}(\bm{\alpha}) \bm{S}),
\end{equation}
and including the positive definiteness of $\bm{\Gamma}(\bm{\alpha})$ as constraint, the likelihood estimation reads as
\begin{equation}
\begin{aligned}
    \label{eq:repara_likelihood}
    \hat{\bm{\alpha}}(\mathcal{D}) = & \argmax\limits_{\bm{\alpha}}\ \mathcal{L}_{\mathcal{D}}(\bm{\alpha}) \\ \mathrm{s.t.}\ \bm{\alpha} \in & \hspace{0.5em} \mathcal{G}_{\alpha} =
 \{\bm{\alpha} \in \mathbb{R}_{++} \times \mathbb{C}^{P-1} : \bm{\Gamma}(\bm{\alpha}) \in \mathbb{S}^P_{++}\}
\end{aligned}
\end{equation}
with $\hat{\bm{\Gamma}}(\mathcal{D}) = \bm{\Gamma}(\hat{\bm{\alpha}}(\mathcal{D}))$. Due to the impractical formulation of $\mathcal{G}_{\alpha}$ for applying standard optimization algorithms, the \ac{OP} for the exact log-likelihood in \eqref{eq:repara_likelihood} is generally ill-posed \cite{McWhorter1995}\cite[Sec. 5.9]{Hamilton1994}. Thus, we reformulate $\mathcal{G}_{\alpha}$ in the following to end up with a well-defined \ac{OP}.

\section{Related Work}

In this section, we give a summary about other methods used for estimating Toeplitz structured \acp{CM}. Generally, these estimators can be classified into likelihood estimators, banding and tapering estimators, and shrinkage estimators.

\subsection{Likelihood Estimators}
\label{sec:likelihood_estimators}
In existing literature, the proposed likelihood estimators for Toeplitz structured \acp{CM} aim to identify the optimizer
\begin{equation}
    \hat{\bm{C}} = \argmax_{\bm{C} \in \mathcal{C}_{\mathrm{T}}}\ - \log \det \bm{C} - \tr(\bm{C}^{-1}\bm{S})
\end{equation}
with $\mathcal{C}_{\mathrm{T}}$ containing all positive semidefinite Hermitian Toeplitz structured matrices and, thus, they estimate directly the \ac{CM}. The work in \cite{Miller1987} builds on the insights from \cite{Burg1982} and interprets the observed $P$-dimensional samples as incomplete observations from a $G$-periodic process with a $G \times G$ circulant \ac{CM}, with $G > P$. As circulant matrices can be diagonalized by the \ac{DFT} matrix $\bm{F}_G$, the Gaussian likelihood enables a closed-form optimization for circulant \acp{CM}, given by
\begin{equation}
    \label{eq:dft_estimator}
    \hat{\bm{C}}_{\mathrm{Circ}} = \bm{F}_G^{\operatorname{H}}\operatorname{diag}(\bm{F}_G \bm{S}_{\mathrm{C}} \bm{F}_G^{\operatorname{H}}) \bm{F}_G.
\end{equation}
However, the $G \times G$ \ac{SCM} $\bm{S}_{\mathrm{C}}$ is intractable due to only incomplete subsamples of dimension $P$ being observed. To address this problem, an iterative \ac{EM} algorithm is derived based on the work in \cite{Dempster1977}, which consists of repeating the following step:
\begin{equation}
\begin{aligned}
    \hat{\bm{\Sigma}}^{(t+1)} =  \operatorname{diag}(\hat{\bm{\Sigma}}^{(t)}\tilde{\bm{F}}_G\hat{\bm{C}}_p^{(t)-1}\bm{S}\hat{\bm{C}}_p^{(t)-\operatorname{H}}\tilde{\bm{F}}_G^{\operatorname{H}} + \hat{\bm{\Sigma}}^{(t)} - \\ \hat{\bm{\Sigma}}^{(t)}\tilde{\bm{F}}_G\hat{\bm{C}}_p^{(t)-1}\tilde{\bm{F}}_G^{\operatorname{H}}\hat{\bm{\Sigma}}^{(t)\operatorname{H}}),
\end{aligned}
\end{equation}
where $\hat{\bm{\Sigma}}^{(t)}$ is the estimated $G \times G$ Fourier transformed circulant \ac{CM} in the $t$th iteration, $\tilde{\bm{F}}_G$ is the left $G \times P$ submatrix of $\bm{F}_G$, $\bm{S}$ is the $P \times P$ \ac{SCM} of the observations, and $\hat{\bm{C}}_p^{(t)}$ is the left upper $P \times P$ submatrix of $\bm{F}_G\hat{\bm{\Sigma}}^{(t)}\bm{F}_G^{\operatorname{H}}$ and, thus, an estimator of the $P \times P$ Toeplitz structured \ac{CM} in the $t$th iteration. After reaching a termination criterion in the $T$th iteration, $\hat{\bm{C}}^{(T)}_P$ is considered as the estimate of the \ac{CM} and will be denoted by $\hat{\bm{C}}_{\mathrm{EM}}$.

Finally, by assuming $G = P$ and substituting the \ac{SCM} $\bm{S}$ for $\bm{S}_{\mathrm{C}}$, the closed-form likelihood estimator $\hat{\bm{C}}_{\mathrm{Circ}}$ in \eqref{eq:dft_estimator} is considered as another likelihood estimate.

\subsection{Banding and Tapering Estimators}
\label{sec:banding_tapering}
Initially, the method of banding and tapering was applied to the \ac{SCM} to reduce its \ac{MSE} \cite{Bickel2008B}. In both cases, the operation consists of an elementwise multiplication of the \ac{SCM} $\bm{S}$ with a specific mask matrix $\bm{M}$, i.e.,

\begin{equation}
    \label{eq:elementwise_mult}
    \hat{\bm{C}} = \bm{M} \odot \bm{S}.
\end{equation}
In case of banding, the mask matrix $\bm{M}$ is given by
\begin{equation}
    \label{eq:banding_mask}
    \bm{M}_{ij} = \chi(|i-j| \leq k_{\mathrm{B}}),\end{equation}
with $k_{\mathrm{B}}$ being a hyperparameter and $\chi(\cdot)$ being the indicator function. Tapering, on the other hand, is a generalization of banding, in which we allow a windowing function $g(\cdot)$ of choice determining the entries $\bm{M}_{ij}$, i.e.,
\begin{equation}
    \label{eq:tapering_mask}
    \bm{M}_{ij} = g(|i-j|).
\end{equation}
The work in \cite{Cai2013} investigates banding and tapering for estimating Toeplitz structured \acp{CM} by exchanging the \ac{SCM} with the \ac{SAVG} $\bm{S}_{\mathrm{avg}}$, i.e., 
\begin{equation}
    \label{eq:sCov_avg}
    \bm{S}_{\mathrm{avg}} = \frac{\tr(\bm{S})}{P}\operatorname{\mathbf{I}} + \sum_{q=1}^{P-1} \frac{\tr(\bm{S}\operatorname{\mathbf{E}}^q)}{P-q}(\operatorname{\mathbf{E}}^q)^{\operatorname{T}} + \frac{\overline{\tr(\bm{S}\operatorname{\mathbf{E}}^q)}}{P-q}\operatorname{\mathbf{E}}^q.
\end{equation}
However, since $\bm{S}_{\mathrm{avg}}$ is not necessarily positive semidefinite, neither the tapered nor the banded estimator $\hat{\bm{C}}$ in \eqref{eq:elementwise_mult} are guaranteed to possess this property. Various options exist to address this problem. One possibility is described in \cite{Cai2013}, which is based on a modification of the corresponding spectral density. However, these methods may yield rank deficient estimators, making them unsuitable for estimating the \ac{ICM}. 
In fact, \cite{Cai2013} proposes an \ac{ICM} estimator, which is optimal with respect to its convergence rate. However, this estimator suffers from requiring genie knowledge about the true \ac{CM}. Moreover, it neglects the observations and outputs the identity matrix as \ac{ICM} estimate if the \ac{CM} tapering estimator contains an eigenvalue below a certain threshold, which leads to arbitrarily large errors for particular estimates. 
To our knowledge, no practical method has been proposed in the literature for constructing a well-conditioned banding or tapering based estimator of the \ac{ICM}. For tuning the hyperparameter, \cite{Bickel2008B} employs $k$-fold cross validation.
However, the case of having one single sample, i.e., $N = 1$, renders $k$-fold cross validation infeasible and was investigated in \cite{Wu2009} and \cite{Murry2010}. Instead of using the unbiased estimator $\bm{S}_{\mathrm{avg}}$ for banding or tapering, the authors consider a biased version $\bm{S}_{\mathrm{acf}}$, in which the summation over all entries along a \ac{SCM}'s diagonal is divided by the sample's dimension $P$, instead of the number of entries. This is usually referred to as sample autocovariance function \cite[Sec. 1.4.4]{Brockwell2016}. Additionally, the authors in \cite{Wu2009} and \cite{Murry2010} introduce tailored methods for hyperparameter tuning. 
\subsection{Shrinkage Estimator}
\label{sec:shrinkage}
The idea of shrinkage estimators is to utilize the bias-variance trade-off and to intentionally incorporate bias by building a convex combination of the unbiased \ac{SCM} with a biased target matrix $\bm{T}$ \cite{Ledoit2004}. The estimator is given by
\begin{equation}
    \hat{\bm{C}}_{\mathrm{shrink}} = (1 - \rho) \bm{S} + \rho \bm{T}
\end{equation}
with shrinkage coefficient $\rho \in [0,1]$. By minimizing the \ac{MSE} with respect to $\rho$, i.e.,
\begin{equation}
    \rho_{\mathrm{opt}}(\bm{T},\bm{C}) = \argmin_{\rho \in [0,1]} \E[\| (1 - \rho) \bm{S} + \rho \bm{T} - \bm{C} \|_{\mathrm{F}}^2],
\end{equation}
one can find a value, which optimally leverages the bias-variance trade-off. Since $\rho_{\mathrm{opt}}(\bm{T},\bm{C})$ depends on the true \ac{CM} $\bm{C}$, it is intractable and must be estimated. In \cite{ChenWiesel2010}, an estimator $\hat{\rho}$ is proposed for the case of \ac{i.i.d.} Gaussian data samples and a scaled identity as target matrix, i.e., $\bm{T}_{\mathrm{I}} = \frac{\tr(\bm{S})}{P}\operatorname{\mathbf{I}}$. The work in \cite{Lancewicki2014} introduces $\bm{S}_{\mathrm{avg}}$, defined in \eqref{eq:sCov_avg}, as a specific target matrix for Toeplitz structured \acp{CM}. Since $\bm{S}_{\mathrm{avg}}$ is also unbiased, taking it as a target does not introduce bias, but still decreases the variance error due to the utilization of structural knowledge. The work in \cite{Zhang2019} considers another Toeplitz structured target $\bm{T}_{\operatorname{H}}$, which is constant in all off-diagonal entries, and thus, introduces a bias. More precisely,
\begin{equation}
    \label{eq:shrinkage_biased_target}
    \bm{T}_{\operatorname{H}} = \frac{\tr(\bm{S})}{P}\operatorname{\mathbf{I}} + \frac{\tr(\bm{S}\bm{H})}{P(P-1)}\bm{H}
\end{equation}
with $\bm{H} = \bm{1}\bm{1}^{\operatorname{T}} - \operatorname{\mathbf{I}}$ and $\bm{1}$ being the all-ones vector. Due to the positive definiteness of $\bm{T}_{\operatorname{H}}$, the resulting shrinkage estimator is \ac{PD}. In our simulations, we take the proposed shrinkage coefficient estimators $\hat{\rho}$ in \cite{Zhang2019} for both targets $\bm{S}_{\mathrm{avg}}$ and $\bm{T}_{\operatorname{H}}$.

\section{Positive Definiteness Constraints for the \ac{GS} Decomposition}
\label{sec:psd_constraints}

The goal of this section is to bring \eqref{eq:repara_likelihood} into the standard form of \ac{OP}s, cf. \cite[Sec. 1.1]{Boyd2004}:
\begin{equation}
    \max\limits_{\bm{\alpha}} f(\bm{\alpha})\ \ \mathrm{s.t.}\ g_i(\bm{\alpha}) \leq 0,\ i = 1,\ldots,m
\end{equation}
with  $\bm{\alpha} \in \mathbb{R} \times \mathbb{C}^{P-1}$.
Since $\alpha_0$ has to be strictly positive, we introduce a small positive constant $\epsilon_0$ serving as a lower bound for $\alpha_0$ and resulting in the inequality constraint
\begin{equation}
    g_1(\bm{\alpha}) = \epsilon_0 - \alpha_0.
\end{equation}
It remains to reformulate the condition of $\bm{\Gamma}(\bm{\alpha})$ to be \ac{PD}, which is explained in the following. 
\subsection{Eigenvalue Constraints}
\label{sec:eigenvalue}
Various approaches exist to include the constraint of positive definiteness into an \ac{OP}. The geometric approach of iteratively projecting onto $\mathbb{S}^P_{++}$, cf. \cite[Sec. 8.1.1]{Boyd2004}, is not applicable since our \ac{OP} is conducted directly over the domain of $\bm{\alpha}$. Reparameterizing the \ac{OP} using the Cholesky decomposition of \ac{PD} matrices \cite{Benson2003} also cannot be used due to the restriction to the \ac{GS} parameterization \eqref{eq:GS_formula}. 
Alternatively, we can ensure positive definiteness by enforcing positive eigenvalues of $\bm{\Gamma}(\bm{\alpha})$ \cite{Benson2003}, i.e.,
\begin{equation}
\begin{aligned}
\label{eq:eig_opt}
    \hat{\bm{\alpha}}(\mathcal{D}) = & \argmax\limits_{\bm{\alpha}}\ \mathcal{L}_{\mathcal{D}}(\bm{\alpha}) \\  &\quad\ \mathrm{s.t.} \quad \epsilon_0 - \alpha_0 \leq 0 \\ & \quad \quad \quad \ \
    \epsilon_i - \lambda_i(\bm{\Gamma}(\bm{\alpha})) \leq 0,\ i = 1,\ldots,P
\end{aligned}
\end{equation}
with $\lambda_i(\bm{\Gamma}(\bm{\alpha}))$ being the $i$th eigenvalue of $\bm{\Gamma}(\bm{\alpha})$.
The eigenvalues of a Hermitian matrix are differentiable \cite[Th. 6.8]{Kato1966}, rendering the \ac{OP} solvable with, e.g., an interior point method. We introduce small positive constants $\epsilon_i$ for each eigenvalue to enforce strict positivity. 
However, the eigenvalue decomposition in each step takes $\mathcal{O}(P^3)$ \acp{FLOP} and the number of nonlinear inequality constraints scales with $P$. These observations render the estimator to be computationally infeasible in high-dimensional settings and, thus, more advanced constraint sets are needed. 
\subsection{Frobenius-Based Constraint}
\label{sec:frobenius}
By inserting the \ac{GS} decomposition \eqref{eq:GS_formula} into the definition of positive definiteness 
and a few reformulations, we obtain the following equivalence.
\begin{lemma}
\label{lem:frob}
If $\bm{\Gamma}(\bm{\alpha})$ satisfies the \ac{GS} decomposition in \eqref{eq:GS_formula} and $\mathcal{A} \subseteq \mathbb{R}_{++} \times \mathbb{C}^{P-1}$, then
\begin{align}
\left(\bm{\Gamma}(\bm{\alpha})\in \mathbb{S}^P_{++}\ \mathrm{for}\ \mathrm{all}\ \bm{\alpha} \in \mathcal{A}\right) \! \! \iff \! \! \left(\max_{\bm{\alpha} \in \mathcal{A}} \|\bm{Z}^{\operatorname{H}}\bm{B}^{-\operatorname{H}}\|_2^2 < 1\right)
\end{align}
with $\bm{B}$ and $\bm{Z}$ given by \eqref{eq:B_def} and \eqref{eq:Z_def}, respectively. 
\end{lemma}
\begin{proof}
See Appendix \ref{proof_lemma_frob}.
\end{proof}
Lemma \ref{lem:frob} can be utilized to reformulate the \ac{OP} \eqref{eq:eig_opt} to just having two inequality constraints independently of the samples' dimension $P$. The spectral norm, however, is not differentiable everywhere \cite[Th. 1.1]{Lewis1994}. One method to overcome this issue is to apply a subgradient method for finding a local optimum. Alternatively, we introduce a differentiable upper bound for the spectral norm, which ensures the differentiability of the constraint. More precisely, we leverage the relationship that the Frobenius norm of a matrix is always greater than or equal to its spectral norm. This allows us to deduce the following implication:
\begin{equation}
    \label{eq:impli}
    \left(\max_{\bm{\alpha} \in \mathcal{A}} \|\bm{Z}^{\operatorname{H}}\bm{B}^{-\operatorname{H}}\|_F^2 < 1\right) \Rightarrow \left(\bm{\Gamma}(\bm{\alpha})\in \mathbb{S}^P_{++}\ \mathrm{for}\ \mathrm{all}\ \bm{\alpha} \in \mathcal{A}\right).
\end{equation}
By utilizing \eqref{eq:impli}, we reformulate the \ac{OP} \eqref{eq:repara_likelihood} while preserving the differentiability of the constraint
\begin{equation}
\begin{aligned}
\label{eq:frob_opt}
    \hat{\bm{\alpha}}(\mathcal{D}) = & \argmax\limits_{\bm{\alpha}}\ \mathcal{L}_{\mathcal{D}}(\bm{\alpha})\\ & \quad \ \mathrm{s.t.} \quad \epsilon_0 - \alpha_0 \leq 0 \\ & \quad \quad \quad \ \ \|\bm{Z}^{\operatorname{H}}\bm{B}^{-\operatorname{H}}\|_F^2 - 1 + \epsilon_f \leq 0
\end{aligned}
\end{equation}
with $\epsilon_f$ being a small positive constant enforcing $\|\bm{Z}^{\operatorname{H}}\bm{B}^{-\operatorname{H}}\|_F^2$ to be strictly smaller $1$. Both, the inversion of a triangular Toeplitz matrix as well as the multiplication of two upper (lower) triangular Toeplitz matrices preserve the triangular Toeplitz structure \cite{Lin2004,Kucerovsky2016}. Thus, $\|\bm{Z}^{\operatorname{H}}\bm{B}^{-\operatorname{H}}\|_F^2$ is fully determined by the entries of $\bm{B}^{-1}\bm{z}_1$, where $\bm{z}_1$ is the first column of $\bm{Z}$. It takes $\mathcal{O}(P^2)$ operations to compute $\bm{B}^{-1}\bm{z}_1$, leading to a significantly reduced complexity to evaluate the constraint. While the eigenvalue constraints in \eqref{eq:eig_opt} represent an equivalent reformulation of $\mathcal{G}_{\alpha}$ in \eqref{eq:repara_likelihood}, the constraint set in \ac{OP} \eqref{eq:frob_opt} is a subset of $\mathcal{G}_{\alpha}$. A smaller constraint set has a stronger regularization effect and generally decreases the estimator's variance. At the same time, it potentially introduces a bias, whose effect can be either beneficial or disadvantageous depending on how well the smaller constraint set can cover the estimator's most relevant features. In our simulations, we either observed similar or even better performance by taking the constraint set in \eqref{eq:frob_opt} instead of \eqref{eq:eig_opt}, rendering it well suited for the problem at hand.

\subsection{Box Constraints}
\label{sec:box_constraints}

Convexity is generally a desirable property for the constraint set because it allows for projecting onto the set by solving a convex \ac{OP} \cite[Sec. 8.1.1]{Boyd2004}. Box constraints further simplify projections, requiring only $\mathcal{O}(P)$ \acp{FLOP}. In the following, we derive box constraints constituting a subset of $\mathcal{G}_{\alpha}$ and, thus, imposing a stronger regularization. However, they offer an additional means to control this effect by incorporating additional prior knowledge, which will be discussed in more detail in Section \ref{sec:tuning_bounds} and, thus, their advantages extend beyond computational efficiency.
We start with the constraint on the Frobenius norm in \eqref{eq:frob_opt}. The matrix $\bm{B}$ is a lower triangular Toeplitz matrix for which the inverse can be computed by means of the generalized Fibonacci polynomials \cite{Sahin2018}. By utilizing this connection, we derive the following relation.

\begin{lemma}
\label{lem:g}
Let $\bm{B}$ and $\bm{Z}$ be defined according to \eqref{eq:B_def} and \eqref{eq:Z_def}, respectively, then
\begin{equation}
    \|\bm{Z}^{\operatorname{H}}\bm{B}^{\mathrm{-H}}\|_F^2 = \sum_{d=1}^{P-1} (P-d) |g_d|^2
\end{equation}
with 
\begin{equation}
    \label{eq:g_d}
    g_d = \sum_{j=1}^d \frac{\alpha_{P-j}}{\alpha_0}F_{d-j}\left(-\frac{\overline{\bm{\alpha}_{\geq1}}}{\alpha_0}\right)\ \mathrm{for}\ \mathrm{all}\ d = 1,\ldots,P-1,
\end{equation}
and $F_{i}(\cdot)$ being the special case of the generalized Fibonacci sequence in \eqref{eq:fibonacci_definition}.
\end{lemma}
\begin{proof}
See Appendix \ref{proof_lemma_g}.
\end{proof}
Motivated by the observation that \eqref{eq:g_d} only depends on $\bm{\alpha}$ through fractions of the form $\alpha_i/\alpha_0$ ($i>0$), we introduce the following box constraints
\begin{equation}
\label{eq:K_bound}
    \left|\frac{\alpha_i}{\alpha_0}\right| \leq K_i,\ K_i > 0\ \mathrm{for}\ \mathrm{all}\ i = 1,\ldots,P-1.
\end{equation}
By means of Lemma \ref{lem:g} and the bounds in \eqref{eq:K_bound}, we derive the following Theorem, which guarantees the existence of positive definiteness enforcing box constraints. 
\begin{theorem}
\label{th:theorem_K}
    Let $\bm{\Gamma}(\bm{\alpha})$ satisfy the \ac{GS} decomposition in \eqref{eq:GS_formula} and let $\tilde{\mathcal{A}}$ be defined as 
    \begin{equation}
    \label{eq:A_set_definition}
    \tilde{\mathcal{A}} =\! \{\bm{\alpha} \in \mathbb{R}_{++} \times \mathbb{C}^{P-1}\! : |\alpha_i| \leq K_i \alpha_0\ \mathrm{for}\ \mathrm{all}\ i = 1,\ldots,P-1\}.
\end{equation}
Moreover, let $\bm{K} = [K_1,\ldots,K_{P-1}]^{\operatorname{T}}$ and $B(\bm{K})$ satisfy
\begin{equation}
    \label{eq:B_of_K}
    B(\bm{K}) = \sum_{d=1}^{P-1}(P-d)\left(\sum_{j=1}^{d} K_{P-j} F_{d-j}(\bm{K})\right)^2.
\end{equation}
Then,
\begin{equation}
    \exists \bm{K} \in \mathbb{R}_{++}^{P-1}: B(\bm{K}) < 1
\end{equation}
with $B(\bm{K})$ exhibiting the following property:
\begin{equation}
    \label{eq:B_property}
    (B(\bm{K}) < 1) \Rightarrow (\bm{\Gamma}(\bm{\alpha})\in \mathbb{S}^P_{++}\ \mathrm{for}\ \mathrm{all}\ \bm{\alpha} \in \tilde{\mathcal{A}}).
\end{equation}
\end{theorem}
\begin{proof}
See Appendix \ref{proof_theorem_K}.
\end{proof}
The property of $B(\bm{K})$ in \eqref{eq:B_property} provides a means to determine specific values for $\bm{K}$. To do so, we propose introducing a functional dependency between the entries of $\bm{K}$, which not only simplifies the computation but also provides the possibility to incorporate additional prior knowledge, as discussed in more detail in Section \ref{sec:tuning_bounds}. However, to ensure the existence of the resulting box constraints, we limit ourselves to a specific class of functions. The formal statement of this restriction is provided in the following Lemma.

\begin{lemma}
\label{lem:lemma_f}
Let $f$ be a function with the following properties
\begin{itemize}
    \item[a)] $f(\eta,i): \mathbb{R}_+ \times \mathbb{N} \rightarrow \mathbb{R}_+$,
    \item[b)] $f(0,i) = 0\ \mathrm{for}\ \mathrm{all}\ i \in \mathbb{N}$,
    \item[c)] $f(\eta,i)$ is continuous and monotonically increasing in $\eta$.
\end{itemize}
Additionally, let $\bm{f}(\eta) = [f(\eta,1),f(\eta,2),\ldots,f(\eta,P-1)]^{\mathrm{T}}$, $\bm{\Gamma}(\bm{\alpha})$ satisfy the \ac{GS} decomposition in \eqref{eq:GS_formula}, $\tilde{\mathcal{A}}$ be defined according to \eqref{eq:A_set_definition} with 
\begin{equation}
    \label{eq:K_func_dep}
    K_i = f(\eta,i)\ \mathrm{for}\ \mathrm{all}\ i=1,\ldots,P-1,
\end{equation}
and $B_f(\eta)$ be defined as the concatenation of $B(\cdot)$ in \eqref{eq:B_of_K} and $\bm{f}(\eta)$, i.e., $B_f(\eta) = B(\bm{f}(\eta))$. Then
\begin{equation}
    \exists \eta \in \mathbb{R}_{++}: B_f(\eta) < 1
\end{equation}
with $B_f(\eta)$ exhibiting the following property 
\begin{equation}
    (B_f(\eta) < 1) \Rightarrow (\bm{\Gamma}(\bm{\alpha})\in \mathbb{S}^P_{++}\ \mathrm{for}\ \mathrm{all}\ \bm{\alpha} \in \tilde{\mathcal{A}}).
\end{equation}
\end{lemma}
\begin{proof}
See Appendix \ref{proof_lemma_f}.
\end{proof}
Lemma \ref{lem:lemma_f} simplifies the task of determining the vector-valued $\bm{K}$ to deciding for a single value $\eta$ to fully determine the constraint set. We propose to select the largest possible $\eta$ for which $B_f(\eta) < 1$, since this leads to the largest possible constraint set within all possible sets satisfying \eqref{eq:A_set_definition} and \eqref{eq:K_func_dep}. The function $B_f(\eta)$ is constructed to be monotonically increasing in $\eta$ and, thus, the parameter $\eta$ can be precomputed by bisection, i.e., we search for a $\eta$ for which
\begin{equation}
     1 - \epsilon_\eta \leq B_f(\eta) < 1
\end{equation}
holds with $\epsilon_\eta$ being a small positive tolerance parameter. In the case of no prior knowledge, the function $f$ could be chosen to constrain all parameters equally, i.e., 
\begin{equation}
    f(\eta,i) = \eta.
\end{equation} 
In case of prior knowledge, however, it is possible to choose more sophisticated functions $f$, which will be discussed in Section \ref{sec:tuning_bounds}. For solely real-valued $\alpha_i$ ($i > 0$), we obtain the following \ac{OP} by integrating the box constraints \eqref{eq:K_bound} into \eqref{eq:repara_likelihood}:
\begin{equation}
\begin{aligned}
\label{eq:box_opt}
    \hat{\bm{\alpha}}(\mathcal{D}) = & \argmax\limits_{\bm{\alpha}}\ \mathcal{L}_{\mathcal{D}}(\bm{\alpha}) \\& \! \! \! \! \mathrm{s.t.} \ \ \epsilon_0 - \alpha_0 \leq 0 \\ & \quad \
     \alpha_i - K_i\alpha_0 \leq 0\ \mathrm{for}\ \mathrm{all}\ i = 1,\ldots,P-1 \\ & \quad
     - \alpha_i - K_i\alpha_0 \leq 0\ \mathrm{for}\ \mathrm{all}\ i = 1,\ldots,P-1,
\end{aligned}
\end{equation}
which can be solved by, e.g., a projected gradient descent. To ensure the box constraints \eqref{eq:K_bound} in the complex-valued case, one possibility is to bound the imaginary and real part of $\alpha_i/\alpha_0$ in absolute value equally by $K_i/2$.

\section{Relation to Autoregressive Processes And Hyperparameter Tuning}
\label{sec:hyperparameter_section}
\subsection{Relation to Autoregressive Processes}
\label{sec:relation}
Let 
\begin{equation}
    \label{eq:AR_process}
    X_t = \sum_{i=1}^{w}a_i X_{t-i} + e_t
\end{equation}
be an \ac{AR}$(w)$ process of order $w$ with \ac{AWGN} $e_t \sim \mathcal{N}_{\mathbb{C}}(0,\sigma^2)$. The second moments of \ac{AR} processes can be put into relation with its parameters by the Yule-Walker equations
\begin{equation}
    \label{eq:yule_walker}
    c(j) = \sum_{i=1}^{w} a_i c(i-j) + \delta(j)\sigma^2\ \mathrm{for}\ \mathrm{all}\ j = 0,\ldots,w
\end{equation}
with $c(j)$ being the autocovariance sequence \cite[Sec. 2.1.3]{Kirchgaessner2007}. By rewriting \eqref{eq:yule_walker} as a system of linear equations in matrix vector form with $\bm{C}_{ij} = c(i-j)$ being the \ac{CM} of a $w+1$-dimensional sample of the process and by inserting the \ac{GS} parameterization \eqref{eq:GS_formula} for $\bm{C}^{-1}$, we observe the following relation
\begin{align}
    \label{eq:relation_GS_A}
    \alpha_0 = &  \frac{1}{\sigma^2},\ \ 
    %\label{eq:relation_GS_A_2}
    \alpha_i = - \frac{a_i}{\sigma^2}\ \mathrm{for}\ \mathrm{all}\ i = 1,\ldots,w.
\end{align}
The relationship between the \ac{GS} parameterization and the parameters of an \ac{AR} process was initially explored in \cite{Kailath1975}, establishing an interpretability of the former in terms of the latter. Building upon this interpretability, we leverage it to propose techniques for tuning the hyperparameters of our proposed estimators.
\subsection{Order Tuning}
\label{sec:order_tuning}
Irrespective of the chosen constraint set from Section \ref{sec:psd_constraints}, we propose using the \ac{BIC} for tuning the number of nonzero $\alpha_i$ ($i > 0)$.  The \ac{BIC} is commonly utilized to select the statistical model for likelihood estimation out of a set of potential models, whether there is only $N=1$ or $N > 1$ observations \cite[Sec. 5.5.2]{Brockwell2016},\cite[Sec. 4.4.1]{Bishop2007}. For our setup and $N > 1$, the \ac{BIC} reads as

\begin{equation}
\label{eq:BIC_choice}
    i_{\mathrm{BIC}} = \argmin_{i \in \{1,\ldots,P\}} i \log N - \mathcal{L}_i(\hat{\bm{\alpha}}_i(\mathcal{D}))
\end{equation}
with $\mathcal{L}_i(\hat{\bm{\alpha}}_i(\mathcal{D}))$ being the log-likelihood evaluated at the estimator $\hat{\bm{\alpha}}_i(\mathcal{D}) \in \mathbb{R}^{i}$ obtained by optimizing the \ac{OP} of choice in Section \ref{sec:psd_constraints}. To evaluate the $i$th log-likelihood, we set the last $P-i$ parameters $\{\alpha_h\}_{h=i}^{P-1}$ in the \ac{GS} decomposition \eqref{eq:GS_formula} to zero and optimize solely over the remaining $i$ parameters $\{\alpha_h \}_{h=0}^{i-1}$. Based on the connection between the \ac{GS} parameterization and \ac{AR} processes, cf. Section \ref{sec:relation}, setting only the first $i$ parameters to nonzero is equivalent to fitting an \ac{AR} process of order $i-1$ to the available dataset. Moreover, using the \ac{BIC} for order tuning results in a duality between our estimators and banding estimators, which is explained in Section \ref{sec:duality}. 

\subsection{Box Constraint Tuning}
\label{sec:tuning_bounds}

When employing the box constraints discussed in Section \ref{sec:box_constraints}, the function $f(\eta,i)$, detailed in Lemma \ref{lem:lemma_f}, is an additional hyperparameter and must be determined. The \ac{GS} parameters are proportional to those of an \ac{AR} process (cf. Section \ref{sec:relation}), so constraining the former implies constraining the latter. Consequently, by fixing a functional dependency between the bounds for the parameters of the \ac{GS} decomposition, we implicitly assume the same functional dependency between the bounds for the parameters of the \ac{AR} process. Therefore, by setting $f(\eta,i)$, we effectively exclude certain parameter constellations of the \ac{AR} process, which we fit to the given dataset. 

The way this insight can be incorporated in determining $f(\eta,i)$ depends on the degree of prior knowledge. If, e.g., the concrete functional dependency of the underlying \ac{AR} process is pre-known, we can decide for one single $f(\cdot,\cdot)$. However, in cases where only limited prior knowledge is given or it is unknown if the true process exhibits a perfect \ac{AR} structure, relying on one $f(\cdot,\cdot)$ might be too restrictive. In such cases, we propose to decide for a class of functions $\{f_m(\eta_m,i)\}_{m=1}^M$ and to select the one leading to the estimator which maximizes a certain quality measure. One example is the class of exponentially decreasing functions, expressed as
\begin{equation}
f_m(\eta_m,i) = \eta_m \cdot \mathrm{e}^{-\lambda_m i}
\end{equation}
with coefficient $\lambda_m > 0$. This class of functions is motivated by a general understanding of typical \ac{AR} processes, where the parameters tend to decrease with growing $i$. One can take the likelihood value of the resulting estimators as quality measure. Alternatively, one can also apply $k$-fold cross validation \cite[Sec. 1.3]{Bishop2007}. 

\section{Conditional Likelihood Estimation}
\label{sec:approximate_likelihood}

In Section \ref{sec:problem_formulation}, we introduced the exact unconditioned log-likelihood $\mathcal{L}_{\mathcal{D}}(\bm{\alpha})$ for estimating the \ac{GS}, cf. \eqref{eq:loglikelihood_def}, and implicitly the \ac{AR} parameters of any fixed order $w$. On the other hand, it has been shown that maximizing the conditional likelihood of one observed sample, i.e., $N=1$, conditioned on its first $w$ (i.e., the order of the fitted \ac{AR} process) entries comes with the benefit of providing a closed form solution of the estimated \ac{AR} parameters \cite[Sec. 5.3]{Hamilton1994}. However, this estimate cannot guarantee the corresponding \ac{ICM} estimate to be \ac{PD} and, thus, no well-conditioned Toeplitz structured \ac{CM} estimator can be directly inferred. In a first step towards combining this insight with our results from Section \ref{sec:psd_constraints} and \ref{sec:hyperparameter_section} to yield a computationally cheap closed form \ac{PD} estimator of the \ac{ICM} as well as the \ac{CM}, we adapt the setting to the complex-valued and more general case of having $N \geq 1$ samples. More precisely, let $\mathcal{D}_C = \{ \bm{x}_{< w}^{(n)} \}_{n=1}^N$ and $\mathcal{D}_L = \{ \bm{x}_{\geq w}^{(n)} \}_{n=1}^N$, then the conditional log-likelihood $\mathcal{L}_{L | C}(\bm{a},\sigma^2)$ of $\mathcal{D}_L$ conditioned on $\mathcal{D}_C$ with respect to the \ac{AR} parameters $\bm{a} = [a_1,\ldots,a_w]^{\operatorname{T}}$ and $\sigma^2$, cf. Section \ref{sec:relation}, reads as \cite[Sec. 5.3]{Hamilton1994}
\begin{equation}
    \label{eq:conditional_log_likelihood}
    \mathcal{L}_{L|C} = - (P-w)\log(\sigma^2) - \frac{\sum_{n=1}^{N}\|\bm{x}^{(n)}_{\geq w} - \bm{\mu}_c^{(n)}(\bm{a})\|_2^2}{\sigma^2 N}
\end{equation}
with the $i$th entry of $\bm{\mu}_c^{(n)}$ being $\sum_{m=1}^{w}a_m x^{(n)}_{i-m}$ and $x^{(n)}_{i-m}$ representing the $(i-m)$th entry of $\bm{x}^{(n)}$. By extending the derivation of the closed form estimator for the \ac{AR} parameters from \cite[Sec. 6.5.1]{Kay1988}, we establish the following relation.

\begin{lemma}
\label{lem:lemma_con}
Let $\mathcal{L}_{L|C}(\bm{a},\sigma^2)$ be defined according to \eqref{eq:conditional_log_likelihood}. Then, the maximizers $\hat{\bm{a}}$ and $\hat{\sigma}^2$ of $\mathcal{L}_{L|C}(\bm{a},\sigma^2)$ satisfy
\begin{equation}
    \label{eq:closed_form_AR}
    \begin{aligned}
        \hat{\bm{a}} = \tilde{\bm{S}}^{-1}_{\geq 1, \geq 1}\tilde{\bm{S}}_{\geq 1,0},\ \ 
        \hat{\sigma}^2 = \frac{1}{P-w}\left( \tilde{\bm{S}}_{0,0} - \tilde{\bm{S}}_{\geq 1,0}^{\operatorname{H}}\hat{\bm{a}} \right),
    \end{aligned}
\end{equation}
where $\tilde{\bm{S}} \in \mathbb{C}^{w+1 \times w+1}$ is defined as
\begin{equation}
    \label{eq:S_tilde}
    \tilde{\bm{S}}_{j,l} = \sum_{t = w}^{P-1} \bm{S}_{t-l,t-j}
\end{equation}
with $\bm{S}$ being the \ac{SCM}.
\end{lemma}
\begin{proof}
See Appendix \ref{proof_lemma_con}.
\end{proof}
The disadvantage of the \ac{AR} parameter estimates in \eqref{eq:closed_form_AR} for estimating the \ac{CM} is the missing guarantee of positive definiteness of the corresponding \ac{ICM}. For increasing dimension $P$, however, the distributions of the exact and conditional likelihood estimator converge to each other \cite[Sec. 5.2]{Hamilton1994}. This motivates the assumption that for very large $P$, the conditional likelihood estimator also provides positive definiteness, which is implicitly assumed in, e.g., \cite{David}. However, there is no mathematical guarantee for finite dimensions, and for small $P$ or close-to-unstable \ac{AR} processes, this assumption does not hold \cite[Sec. 6.5.2]{Kay1988}. To overcome this issue, we propose to utilize the derived box constraints in Section \ref{sec:box_constraints} and to project the corresponding \ac{GS} parameters resulting from \eqref{eq:closed_form_AR} and \eqref{eq:relation_GS_A} onto this constraint set, i.e.,
\begin{equation}
    \label{eq:LS_est}
    \hat{\alpha}_{0,\mathrm{PLS}} = \frac{1}{\hat{\sigma}^2},\ \ \hat{\bm{\alpha}}_{\geq 1,\mathrm{PLS}} = \mathrm{Proj}_{\bm{K}}\left( - \frac{1}{\hat{\sigma}^2}\hat{\bm{a}}\right),
\end{equation}
where $\mathrm{Proj}_{\bm{K}}(\cdot)$ is the projection onto the box constraints defined by \eqref{eq:K_bound}, and $\hat{\sigma}^2$ and $\hat{\bm{a}}$ are given by \eqref{eq:closed_form_AR}. The corresponding \ac{ICM} and \ac{CM} estimate can then be computed from \eqref{eq:GS_formula} and are referred to as \ac{PLStext} estimates. This method has not only the key advantage that it provides a computationally cheap closed-form estimator, which guarantees positive definiteness in any case, but also allows applying the proposed order tuning as well as the box constraint tuning discussed in Section \ref{sec:hyperparameter_section}. On the other hand, it should be noted that this method cannot be theoretically guaranteed to yield a local optimum of either the exact or the conditional likelihood.

\section{Duality to Banding Estimators}
\label{sec:duality}
The introduced hyperparameter tuning method in Section \ref{sec:order_tuning} fits an \ac{AR} process with the most suitable order to the given dataset. In this Section, we discuss this method's regularization effect by considering the resulting imposed \ac{ICM} structure and identify a certain duality to banding estimators, introduced in Section \ref{sec:banding_tapering}.

The $i$,$j$th entry in the \ac{ICM} of a random vector $\bm{z}$ is zero if and only if its $i$th and $j$th entry are uncorrelated given all other entries \cite{FriedmanHastie2018}. Thus, by considering a $P$-dimensional sample from an \ac{AR}($w$) process, with $w < P$, the corresponding \ac{ICM} is banded by a bandwidth of $w$. In consequence, by only allowing $\{\alpha_h\}_{h=0}^{i_{\mathrm{BIC}}-1}$ to be nonzero and due to their relationship to the \ac{AR} parameters in \eqref{eq:relation_GS_A}, our proposed order tuning method in Section \ref{sec:order_tuning} effectively bands the estimated \ac{ICM} by a bandwidth of $i_{\mathrm{BIC}}-1$. In contrast, when considering banding estimators, banding is applied to an estimator of the \ac{CM}, i.e., $\bm{S}_{\mathrm{avg}}$ in \eqref{eq:sCov_avg}. The $i$,$j$th entry in the \ac{CM} of a random vector $\bm{z}$ is zero if and only if its $i$th and $j$th entry are uncorrelated. When considering a \ac{MA} process, defined by $X_t = e_t + \sum_{i=1}^{w} b_ie_{t-i}$ with $e_{\tau}$ being \ac{AWGN} and parameters $\{b_i\}_{i=1}^{w}$, the corresponding \ac{CM} of a $P$-dimensional sample exhibits a band structure of bandwidth $w$. Therefore, it is possible to interpret the process of banding $\bm{S}_{\mathrm{avg}}$ by a bandwidth $w$ as the fitting of a \ac{MA}($w$) process of order $w$ to the given dataset.

As a result, while banding estimators fit a \ac{MA} process to the given dataset and impose a banding structure on the estimated \ac{CM}, our proposed estimators fit an \ac{AR} process and impose a banding structure on the estimated \ac{ICM}. Due to this duality, we can interpret the regularization effect of both estimators in the same but dual manner. Banding estimators reduce the \ac{MSE} in situations where the true \ac{CM} possesses a band structure, i.e., in case of a \ac{MA} process. However, even in cases where the \ac{CM} is not inherently banded, this operation introduces a regularization effect by introducing a bias and leveraging the bias-variance trade-off. The same argumentation holds for our proposed estimators in terms of the \ac{ICM}. 

\section{Computational Complexity}
\label{sec:computational_complexity}
The complexity plays a key role for many applications which require the estimation of \acp{CM}. This can be due to a high dimensionality but also due to the necessity for low-latency real-time computations \cite{Neumann2018}. In the following, we will analyze the complexity of optimizing the exact log-likelihood with the box constraints introduced in Section \ref{sec:box_constraints}. We solve it with a projected gradient descent, which consists of two steps constituting the computational complexity, i.e., determining the search direction and the step size.

\subsection{Determining the Search Direction}

We start with computing the gradient of the objective $\mathcal{L}_{\mathcal{D}}(\bm{\alpha})$ defined in \eqref{eq:loglikelihood_def}.
%\begin{equation}
%    \label{eq:objective}
%    f(\bm{\alpha}) = - \log \det \bm{\Gamma}(\bm{\alpha}) + \tr(\bm{\Gamma}(\bm{\alpha})\bm{S}).
%\end{equation}
The derivative with respect to $\alpha_0$ is given by
\begin{equation}
    \label{eq:grad_0}
    \frac{\partial}{\partial \alpha_0} \mathcal{L}_{\mathcal{D}}(\bm{\alpha}) = \frac{1}{\alpha_0} \tr((\bm{\Gamma}^{-1} - \bm{S})(\bm{B} + \bm{B}^{\operatorname{H}} - \bm{\Gamma})).
\end{equation}
The remaining derivatives are in the real-valued case 
\begin{equation}
    \label{eq:grad_i}
    \frac{\partial}{\partial \alpha_i} \mathcal{L}_{\mathcal{D}}(\bm{\alpha}) = \frac{2}{\alpha_0} \tr((\bm{\Gamma}^{-1} - \bm{S})(\bm{B}(\operatorname{\mathbf{E}}^{i})^{\operatorname{T}} - \bm{Z}(\operatorname{\mathbf{E}}^{P-i})^{\operatorname{T}}))
\end{equation}
and in the complex-valued case
\begin{equation}
\begin{aligned}
    \label{eq:grad_i_complex}
    \frac{\partial}{\partial \alpha_i} \mathcal{L}_{\mathcal{D}}(\bm{\alpha}) = \frac{1}{\alpha_0} \tr((\bm{\Gamma}^{-1} - \bm{S} )^{\operatorname{T}}\overline{\bm{B}}(\operatorname{\mathbf{E}}^{i})^{\operatorname{T}})\\ - \tr(( \bm{\Gamma}^{-1} - \bm{S} )\bm{Z}(\operatorname{\mathbf{E}}^{P-i})^{\operatorname{T}})
    \end{aligned}
\end{equation}
with $i \in \{1\ldots,P-1\}$, $\bm{S}$ being the \ac{SCM}, $\bm{B}$ and $\bm{Z}$ given by \eqref{eq:B_def} and \eqref{eq:Z_def}, and $\operatorname{\mathbf{E}}$ being the shift-down matrix. 
%For complex-valued $\alpha_i$ ($i>0$), the Wirtinger derivative with respect to $\overline{\alpha_i}$ equals \eqref{eq:grad_i} divided by $2$. 
Based on the connection between the \ac{GS} and the \ac{AR} parameters in \eqref{eq:relation_GS_A}, we can compute $\bm{\Gamma}(\bm{\alpha})^{-1}$ by the reversed Levinson algorithm, taking $\mathcal{O}(P^2)$ \acp{FLOP} \cite[Sec. 6.3.4]{Kay1988}. To determine the gradient's complexity, we establish the following Lemmata.

\begin{lemma}
\label{lem:lemma_tr}
Let $\bm{C}_T$ be a Hermitian Toeplitz matrix and let $\bm{D}$ be a lower triangular Toeplitz matrix and $k \in \{0,\ldots,P-1\}$. If $\bm{c} = [c_0,\ldots,c_{P-1}]^{\operatorname{T}}$ and $\bm{d} = [d_0,\ldots,d_{P-1}]^{\operatorname{T}}$ denote the first columns of $\bm{C}_T$ and $\bm{D}$, respectively, then,
\begin{equation}
    \label{eq:trace_term_lem}
    \tr(\bm{C}_T\bm{D}(\operatorname{\mathbf{E}}^k)^{\operatorname{T}}) = \sum_{m=0}^{P-1} \min\{P-k,P-m\} d_m c_{|k-m|}.
\end{equation}
\end{lemma}
\begin{proof}
See Appendix \ref{proof_lemma_tr}.
\end{proof}

\begin{lemma}
\label{lem:lemma_tr2}
Let $\bm{Q}$ be any matrix. Moreover, let $\bm{D}$ be a lower triangular Toeplitz matrix, $k \in \{0,\ldots,P-1\}$ and $\zeta(m) = \min\{P-k-1,P-m-1\}$. If $\bm{d} = [d_0,\ldots,d_{P-1}]^{\operatorname{T}}$ denote the first column of $\bm{D}$, then,
\begin{equation}
    \label{eq:trace_term_lem2}
    \tr(\bm{Q}\bm{D}(\operatorname{\mathbf{E}}^k)^{\operatorname{T}}) = \sum_{m=0}^{P-1} d_m \left( \sum_{j=0}^{\zeta(m)} \bm{Q}_{k+j,m+j} \right).
\end{equation}
\end{lemma}
\begin{proof}
See Appendix \ref{proof_lemma_tr2}.
\end{proof}
A direct conclusion from Lemma \ref{lem:lemma_tr} is that the trace in \eqref{eq:trace_term_lem} and, thus, the term $\tr(\bm{\Gamma}^{-1}(\bm{B}(\operatorname{\mathbf{E}}^{i})^{\operatorname{T}} - \bm{Z}(\operatorname{\mathbf{E}}^{p-i})^{\operatorname{T}}))$ in the real-valued derivative \eqref{eq:grad_i} is computable in $\mathcal{O}(P)$ \acp{FLOP}. To apply Lemma \ref{lem:lemma_tr2}, it is necessary to compute
\begin{equation}
    \sum_{j=0}^{\zeta(m)} \bm{S}_{k+j,m+j}\ \mathrm{for}\ \mathrm{all}\ k,m=0,\ldots,P-1
\end{equation}
with $\bm{S}$ being the \ac{SCM}, which can be precomputed in $\mathcal{O}(P^2)$ \acp{FLOP}. Therefore, computing the remaining part of the real-valued gradient $\tr(\bm{S}(\bm{B}(\operatorname{\mathbf{E}}^{i})^{\operatorname{T}} - \bm{Z}(\operatorname{\mathbf{E}}^{p-i})^{\operatorname{T}}))$ in \eqref{eq:grad_i} can be done in $\mathcal{O}(P)$ \acp{FLOP} according to Lemma \ref{lem:lemma_tr2}. The same argumentation can be applied to the complex-valued case in \eqref{eq:grad_i_complex}. 
Consequently, the computational complexity of computing the derivatives for all nonzero $\alpha_i$ ($i > 0$) is $\mathcal{O}(P^2)$. The derivative with respect to $\alpha_0$ in equation \eqref{eq:grad_0} needs to be calculated once and requires $\mathcal{O}(P^2)$ \acp{FLOP}. Once the gradient is computed, it is projected onto the convex cone of feasible directions to determine the search direction. This projection operation, which guarantees an improving direction \cite{Boyd2004}, takes $\mathcal{O}(P)$ \acp{FLOP}. This results in overall $\mathcal{O}(P^2)$ \acp{FLOP} of computational cost for the search direction.

\subsection{Determining the Step Size}

To determine the step size in the projected gradient descent, we employ the Armijo rule. Since we have calculated the search direction and the gradient beforehand, the only step that adds to the computational complexity is evaluating the objective function at potential candidates for the next iteration. We rewrite the objective function as follows:
\begin{equation}
    \mathcal{L}_{\mathcal{D}}(\bm{\alpha}) = - \log \det \bm{\Gamma}(\bm{\alpha})^{-1} - \tr(\bm{\Gamma}\bm{S})
\end{equation}
where $\bm{\Gamma}(\bm{\alpha})^{-1}$ has also already been computed. To calculate the determinant, we exploit the property that the Cholesky decomposition of a Toeplitz structured \ac{CM} can be computed in $\mathcal{O}(P^2)$ \acp{FLOP} \cite{Bojanczyk1995}. Additionally, determining the \ac{ICM} $\bm{\Gamma}(\bm{\alpha})$ defined in \eqref{eq:GS_formula} also takes $\mathcal{O}(P^2)$ \acp{FLOP}. This is because the $k$th off-diagonal of the product of a $P \times P$ triangular Toeplitz matrix with its Hermitian can be recursively computed, requiring $\mathcal{O}(P-k)$ \acp{FLOP}. As a result, the overall computational complexity of this product is given by the Gauss summation $\mathcal{O}(\sum_{k=1}^P k) = \mathcal{O}(P^2)$. Consequently, $\mathcal{O}(P^2)$ \acp{FLOP} are necessary to evaluate the objective function.

Considering $L_{\mathrm{B}}$ and $T_{\mathrm{B}}$ as the number of iterations for determining the step size and for the projected gradient descent, respectively, the overall computational complexity, without accounting for hyperparameter tuning, is $\mathcal{O}(P^2L_{\mathrm{B}}T_{\mathrm{B}})$ \acp{FLOP} and, thus, scales quadratically with the samples' dimension.%Thus, its complexity scales quadratically with the dimension of the samples.

\section{Simulation Results}
\label{sec:simulation_results}
In the following, we consider different configurations of general \ac{WSS} \ac{AR}, \ac{MA}, and \ac{FBM} processes as ground truth and evaluate the estimators' performance for different numbers $N$ of samples and different dimensions $P$. As a performance measure, we evaluate the \ac{NMSE} $\E[(\|\hat{\bm{A}} - \bm{A}\|_{\mathrm{F}}^2)/\|\bm{A}\|_{\mathrm{F}}^2]$ ($\bm{A} \in \{\bm{C},\bm{\Gamma}\}$) of the estimated \ac{CM} $\hat{\bm{C}}$ ($\mathrm{NMSE}_{\mathrm{C}}$) and estimated \ac{ICM} $\hat{\bm{\Gamma}}$ ($\mathrm{NMSE}_{\mathrm{\Gamma}}$) compared to ground truth by averaging over $N_{\mathrm{runs}} = 500$ Monte Carlo runs. In cases where the considered \ac{CM} estimator does not preserve positive definiteness, we solely evaluate $\mathrm{NMSE}_{\mathrm{C}}$. All estimators we compare ourselves with output an estimate $\hat{\bm{C}}$ of the \ac{CM} rather than an estimate of the inverse. In these cases, the estimated \ac{ICM} $\hat{\bm{\Gamma}}$ is computed by inverting $\hat{\bm{C}}$, if possible. On the other side, our proposed methods output a parameterization of the \ac{ICM} $\bm{\Gamma}$ and, thus, in this situation we derive $\hat{\bm{C}}$ subsequently by inverting $\hat{\bm{\Gamma}}$, i.e., $\hat{\bm{C}} = \hat{\bm{\Gamma}}^{-1}$, which takes $\mathcal{O}(P^2)$ \acp{FLOP}. 
The banding and tapering as well as our proposed estimators require hyperparameter tuning. We utilize $4$-fold cross validation for tuning the hyperparameter of the banding and tapering estimators in all cases with $N \geq 4$. For $N=1$, we apply the hyperparameter tuning from \cite{Wu2009} and \cite{Murry2010}. As windowing function $g(\cdot)$ (cf. Section \ref{sec:banding_tapering}), we choose the trapezoid function from \cite[Sec. 2]{Cai2013} with hyperparameter $k_{\operatorname{T}}$.
For the eigenvalue, Frobenius and box constraint estimators proposed in Section \ref{sec:psd_constraints}, we utilize the \ac{BIC} to tune the order of the fitted \ac{AR} process, cf. Section \ref{sec:order_tuning}. The corresponding likelihood \acp{OP} in \eqref{eq:eig_opt} and \eqref{eq:frob_opt} are solved by an interior point method.
Independently of the process at hand, we use the same class of exponentially decreasing functions $\mathcal{F}$ for the box constraint tuning, cf. Section \ref{sec:tuning_bounds}, i.e,
\begin{equation}
    \label{eq:class_of_functions}
    \mathcal{F} = \{f_m(\eta_m,i) \}_{m=1}^5,
\end{equation}
with $f_m(\eta_m,i) = \eta_m e^{-\lambda_m i}$ and $i \in \{1,\ldots,P-1\}$. The particular values of $\lambda_m$ and $\eta_m$ are given in Table \ref{tab:box_con_functions}. The parameters $\eta_m$ are determined by bisection, cf. Section \ref{sec:box_constraints}. Generally, these parameters are dimension-dependent. However, the exponentially decaying behaviour of the functions contained in $\mathcal{F}$ suppresses the \ac{GS} parameters $\alpha_i$ for large $i$, independently of the chosen dimension. Consequently, if the dimension is sufficiently large, the parameters $\eta_m$ can be approximated to be dimension-independent. With the chosen $\lambda_m$ values, this simplification is valid for dimensions larger 16, and, therefore, we apply it in our simulations. We use the likelihood as performance measure to determine the best-fitting function in $\mathcal{F}$. We did not observe noteworthy differences between using the likelihood and $k$-fold cross validation. Therefore, and because of the disadvantage of leading to a larger computational overhead due to the necessity of computing the estimator $k$ times, we leave out $k$-fold cross validation in our simulations. For hyperparameter tuning, we continually increase the estimator's degrees of freedom, terminate the process after $5$ configurations without increasing hyperparameter performance measure and take the best performing configuration to reduce its computational overhead. The simulation code is publicly available\footnote{\url{https://github.com/beneboeck/toep-cov-estimation}}.
\begin{table}[t]
\centering
 \caption{Chosen parameters of the functions $f_m(\eta_m,i)$ in $\mathcal{F}$.}
 \label{tab:box_con_functions}
\begin{tabular}{|>{\bfseries}p{0.08\columnwidth}|p{0.08\columnwidth}|p{0.08\columnwidth}|p{0.08\columnwidth}|p{0.08\columnwidth}|p{0.08\columnwidth}|p{0.08\columnwidth}}
\hline
$m$ & $1$ & $2$ & $3$ & $4$ & $5$ \\
\hline
$\lambda_m$ & $0.6$ & $1$ & $1.4$ & $1.8$ & $2.2$ \\
\hline
$\eta_m$ & $0.822$ & $1.718$ & $3.055$ & $5.047$ & $8.025$ \\
\hline
\end{tabular}
  \vspace{-0.45cm}
\end{table}

\subsection{Autoregressive and Moving Average Processes}
In Fig. \ref{fig:ARMA_one_processes}, the ´estimated \ac{NMSE} of the estimators is shown for the \ac{WSS} \ac{AR}(1) process $X_t = aX_{t-1} + e_t$ (a) and b)) and the \ac{MA}(1) process $X_t = e_t + b e_{t-1}$ (c) and d)). In all cases, $e_{\tau}$ is \ac{AWGN} with variance $\sigma^2 = 0.64$ and from each process we draw $N=8$ \ac{i.i.d.} samples, where each sample consists of $P = 16$ successive and therefore correlated realizations of the process. In case of an \ac{AR}(1) process, the \ac{ICM} is tridiagonal and the coefficients of the \ac{CM} decrease exponentially along its off-diagonals, while in case of the \ac{MA}(1) process it is vice versa.

In Fig. \ref{fig:ARMA_one_processes} a) and b), the proposed  \ac{Eig}, \ac{Frob}, \ac{PGD} and \ac{PLS} estimators regularize their estimation without necessarily introducing bias, cf. Section \ref{sec:duality}, and lead to the best performance. Note that for $N=8$ and $P=16$, the ordinary conditional likelihood estimator without projection frequently leads to indefinite \acp{ICM} and, thus, cannot be used for \ac{CM} estimation. \ac{PGD} and \ac{PLS} perform slightly better than \ac{Eig} and \ac{Frob}, which is due to the box constraint estimators' stronger regularization. In Fig. \ref{fig:ARMA_one_processes} a), the \ac{band} and \ac{tape} estimators introduce useful bias due to the \ac{CM}'s exponentially decreasing coefficients, but lead to worse performance than our proposed methods. Due to \ac{band}'s and \ac{tape}'s missing guarantee of positive definiteness, they cannot be used for estimating the \ac{ICM} and are not represented in Fig. \ref{fig:ARMA_one_processes} b). Although incorporating the true distribution, the likelihood estimators, i.e., $\hat{\bm{C}}_{\mathrm{Circ}}$ \ac{circ} and $\hat{\bm{C}}_{\mathrm{EM}}$ \ac{em} cannot adapt to the \ac{CM}'s specific characteristics through hyperparameter tuning, which leads to a worse performance in Fig. \ref{fig:ARMA_one_processes} a) and b) compared to the estimators involving hyperparameter tuning. The estimator $\bm{S}_{\mathrm{avg}}$ \ac{avg} in \eqref{eq:sCov_avg} neither involves hyperparameter tuning nor incorporates the distribution and leads in almost all cases in Fig. \ref{fig:ARMA_one_processes} a) to an even worse performance. The \ac{shu} is a convex combination of the \ac{SCM} $\bm{S}$ and \ac{avg} and performs always worse than \ac{avg} in Fig. \ref{fig:ARMA_one_processes} a) due to the imperfect estimation of its convex coefficient. The \ac{shb} performs well for $a=0.1$. This is due to its target matrix $\bm{T}_{\operatorname{H}}$ in \eqref{eq:shrinkage_biased_target}, which weights the main-diagonal and all off-diagonals differently and, thus, is able to suppress all off-diagonals in the case of a diagonally dominant \ac{CM}. However, in all other cases, the equal weighting of all off-diagonal entries introduces a suboptimal bias for the \ac{CM}'s non-negligible off-diagonals and it performs significantly worse. The performance of \ac{shb}'s estimated \ac{ICM} can be explained in the same manner. 
\begin{figure}[t]
    \centering
    %\resizebox{\textwidth}{!}{
  \includegraphics{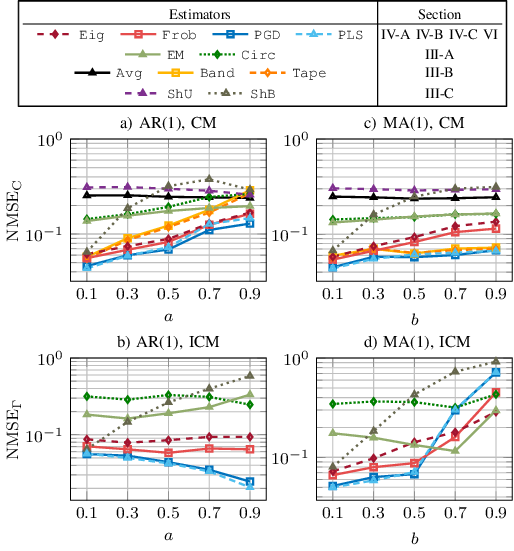}
   \vspace{-6pt}
	\caption{NMSE of the CM and ICM for an AR(1) process with parameter $a$, cf. a) and b), and a MA(1) process with parameter $b$, cf. c) and d). The noise variance is set to $0.64$ in all cases, the dimension $P$ is set to $16$ and the number $N$ of samples equals $8$.}
	\label{fig:ARMA_one_processes}
	\vspace{-0.6cm}
\end{figure}

In case of a \ac{MA}(1) process in Fig. \ref{fig:ARMA_one_processes} c) and d), \ac{band} and \ac{tape} regularize their estimation without necessarily introducing bias, whereas \ac{Eig}, \ac{Frob}, \ac{PGD} and \ac{PLS} introduce bias, cf. Section \ref{sec:duality}. However, the latter are likelihood estimators and, thus, have the advantage of incorporating the true distribution in contrast to the former. By exploiting this additional prior knowledge, our estimators diminish their model mismatch of fitting an \ac{AR} process to the underlying \ac{MA} process, especially for small $b$ in Fig. \ref{fig:ARMA_one_processes} c). 
%These different considerations of the true structure lead to similar performance of these estimators for small $b$ in Fig. \ref{fig:ARMA_one_processes} c). 
For larger $b$, however, the coefficients in the \ac{ICM} decrease slower in absolute value, leading to a stronger bias of \ac{Eig}, \ac{Frob}, \ac{PGD} and \ac{PLS}. This results in a worse performance of \ac{Eig} and \ac{Frob}, whereas the stronger regularization of \ac{PGD} and \ac{PLS} keep their performance in \ac{band}'s and \ac{tape}'s performance range. In terms of estimating the \ac{ICM} in Fig. \ref{fig:ARMA_one_processes} d) for small $b$, the estimators' performance can be explained similarly to their performance in Fig. \ref{fig:ARMA_one_processes} b). For larger $b$, however, \ac{PGD} and \ac{PLS} perform worse than the other likelihood estimators, which can be explained by considering the class of functions $\mathcal{F}$ in \eqref{eq:class_of_functions} used for the box constraint tuning. Generally, $\lambda_m$ in $f_m(\eta_m,i)$ determines how quickly the bound for $|\alpha_i|$ decays in $i$. The smaller we choose $\lambda_m$ in $f_m(\eta_m,i)$, the larger the allowed constraint set is for $\alpha_i$ with large $i$ relative to those with small $i$. However, the smaller $\lambda_m$ gets, the smaller $\eta_m$ has to be to satisfy the constraint in Lemma \ref{lem:lemma_f} and, thus, the smaller the constraint set generally gets. For a \ac{MA}(1) process with large $b$, the true $\bm{\alpha}$ exhibits a slowly decaying behaviour in its entries and only a small gap exists between $\alpha_0$ and the remaining $\alpha_i$ ($i=1,\ldots$). In consequence, none of the functions $f_m(\eta_m,i)$ in $\mathcal{F}$ fits rendering $\mathcal{F}$ ill-suited for this specific situation. Additionally, for larger $b$ the stronger bias of \ac{Frob}, \ac{Eig}, \ac{PGD} and \ac{PLS} worsens their performance in general. From now on, we exclude \ac{Eig} due to the issues concerning its complexity, cf. Section \ref{sec:eigenvalue}.

\begin{figure}[t]
    \centering
    %\resizebox{\textwidth}{!}{
 	\includegraphics{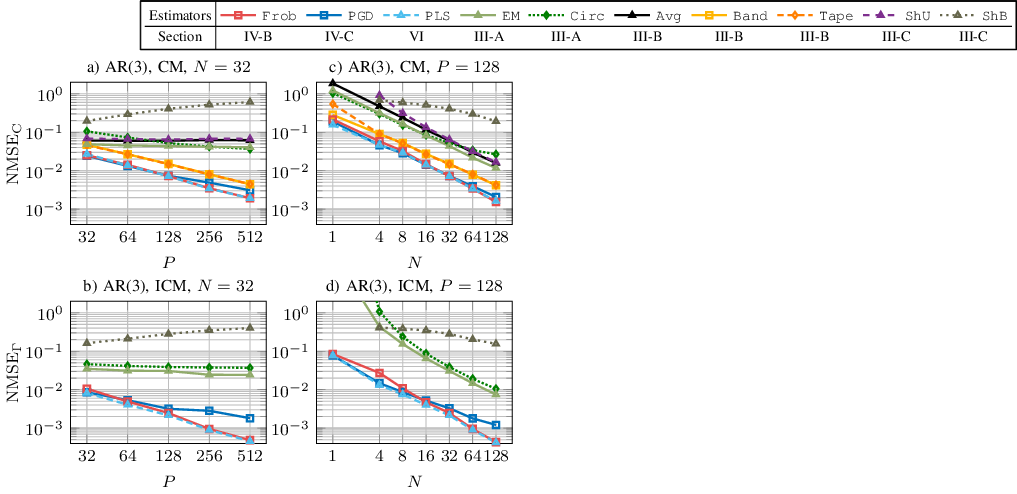}

   \vspace{-6pt}
	\caption{NMSE for an \ac{AR}(3) process ($\sigma^2 = 0.64$, $a_1 = 0.5$, $a_2 = 0.2$, $a_3 = 0.05$) with fixed number of samples $N=32$, cf. a) and b), and with fixed dimension $P=128$, cf. c) and d).}
	\label{fig:AR3}
	\vspace{-0.6cm}
\end{figure}

In Fig. \ref{fig:AR3}, the estimated \ac{NMSE} of the estimators is given in case of the \ac{AR}(3) process $X_t = 0.5 X_{t-1} + 0.2 X_{t-2} + 0.05 X_{t-3} + e_t$ for $N=32$ and varying dimension $P$, cf. a) and b), as well as $P=128$ and varying number $N$ of samples, cf. c) and d). In Fig. \ref{fig:MA3}, the same plots are given in case of the \ac{MA}(3) process $X_t = e_{t} + 0.5 e_{t-1} + 0.2 e_{t-2} + 0.05 e_{t-3}$. In both cases, $e_{\tau}$ is \ac{AWGN} with variance $\sigma^2 = 0.64$. Generally, the performance of the estimators is consistent with the results for the \ac{AR}(1) process and the \ac{MA}(1) process in Fig. \ref{fig:ARMA_one_processes}. For the \ac{AR}(3) process, our proposed estimators \ac{Frob}, \ac{PGD}, and \ac{PLS} generally perform the best due to their regularization without bias and their consideration of the underlying distribution. 

\begin{figure}[t]
    \centering
    \vspace{0.9cm}
    %\resizebox{\textwidth}{!}{
 	\includegraphics{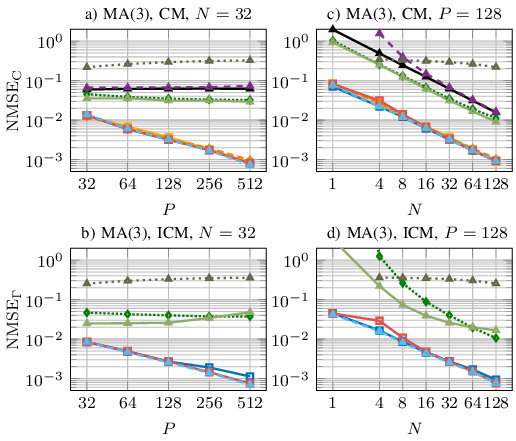}
   \vspace{-18pt}
	\caption{NMSE for a \ac{MA}(3) process ($\sigma^2 = 0.64$, $b_1 = 0.5$, $b_2 = 0.2$, $b_3 = 0.05$) with fixed number of samples $N=32$, cf. a) and b), and with fixed dimension $P=128$, cf. c) and d).}
	\label{fig:MA3}
	\vspace{-0.4cm}
\end{figure} 

For the \ac{CM} of the \ac{MA}(3) process in Fig. \ref{fig:MA3} a) and c), all estimators incorporating hyperparameter tuning perform approximately the same. Although introducing bias, \ac{Frob}, \ac{PGD}, and \ac{PLS} consider the underlying distribution and, thus, are able to equalize \ac{band}'s and \ac{tape}'s better-fitting regularization. For small $N \leq 8$, \ac{PGD} performs slightly better than \ac{Frob}, whereas for large $N \geq 32$ it is vice versa. This behaviour is due to the box constraint estimators' smaller constraint set having a strong regularization effect. Although being beneficial for cases with only a few samples $N$, a strong regularization effect is disadvantageous for larger $N$, where the estimation's variance is small. Generally, \ac{PLS}'s performance approximately matches the better performance of \ac{Frob} and \ac{PGD}. For small $N$, this is due to having the same strong regularization effect as \ac{PGD}, whereas for large $N$ or $P$, finding first the global optimum of the unconstrained conditional likelihood and then projecting onto the box constraints is less restrictive than finding a local optimum within the box constraints.
\subsection{Fractional Brownian Motion}
Another \ac{WSS} process used for modeling long-term dependencies in, e.g., internet traffic, is \ac{FBM}($h$) \cite{Leland1994}. The entries in its corresponding \ac{CM} are defined as
\begin{equation}
    \bm{C}_{ij} = \frac{1}{2} \left( (|i-j| + 1)^{2h} - 2|i-j|^{2h} + (|i-j|-1)^{2h} \right)
\end{equation}
where the so-called Hurst parameter $h$ satisfies $h \in [0.5,1]$. 
In Fig. \ref{fig:FBM}, the estimated \ac{NMSE} of the estimators is evaluated in case of the \ac{FBM}(0.7) process for $N=32$ and a varying dimension $P$, cf. a) and b), as well as $P=128$ and a varying number $N$ of samples, cf. c) and d). The key property of \ac{FBM} is that the entries of the corresponding \ac{CM} exhibit a certain saturation in their decline rate, leading to a full \ac{CM} and \ac{ICM} with non-negligible entries. In consequence, banding the \ac{CM} or \ac{ICM} leads to a larger bias compared to the case of \ac{AR} or \ac{MA} processes, and the estimators incorporating hyperparameter tuning do not outperform the ones without in terms of estimating the \ac{CM} in Fig. \ref{fig:FBM} a). In Fig. \ref{fig:FBM} b), \ac{Frob}, \ac{PGD} and \ac{PLS} notably outperform the compared methods since the true $\bm{\alpha}$ in \ac{FBM} exhibits a strong decrease along its entries in contrast to the \ac{CM}'s off-diagonals. This leads to a better performing bias when banding the \ac{ICM} compared to banding the \ac{CM}. In Fig. \ref{fig:FBM} c), it can be seen that for small $N$, the regularization effect of banding the \ac{CM} and \ac{ICM} is beneficial due to the high estimation's variance. However, for large $N$, this regularization is disadvantageous due to the smaller estimation's variance.
\begin{figure}[t]
\vspace{0.7cm}
    \centering
    %\resizebox{\textwidth}{!}{
   	\includegraphics{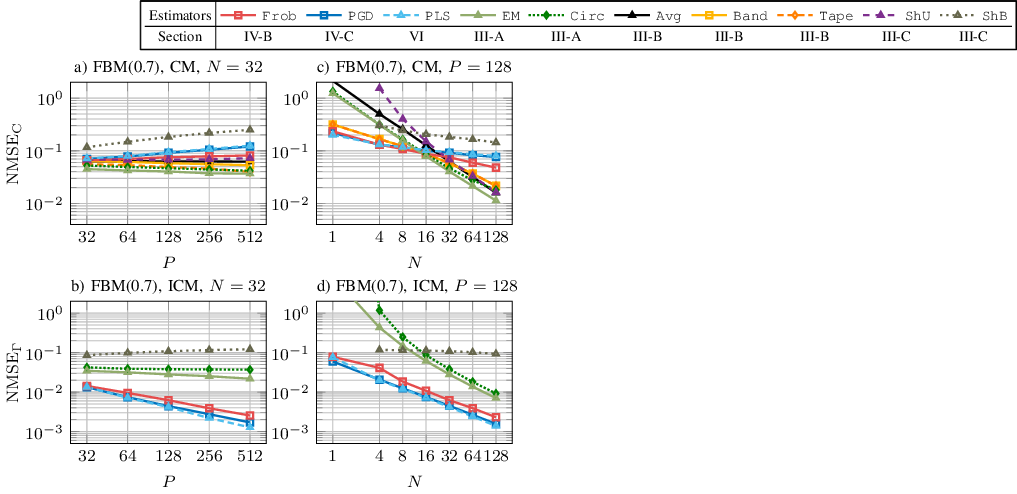}
   \vspace{-17pt}
\caption{NMSE for an \ac{FBM} process ($h = 0.7$) with fixed number of samples $N=32$, cf. a) and b), and with fixed dimension $P=128$, cf. c) and d).}
\label{fig:FBM}
	\vspace{-0.4cm}
\end{figure}
\subsection{Complexity and Runtime Analysis}
In Fig. \ref{fig:complexity}, the complexity of all estimators (right) and the average measured time in seconds of one estimate for variable dimensions $P$ and fixed number $N=64$ of samples (left) are given. In both cases, hyperparameter tuning and the computation of the \ac{SCM} is not considered. The underlying process is an \ac{ARMA}(1,1) process defined as $X_t = 0.7 X_{t-1} + e_t + 0.3e_{t-1}$ with $\sigma^2 = 0.64$. For \ac{band} and \ac{tape} as well as \ac{PGD}, \ac{PLS} and \ac{Frob}, we assume the hyperparameter $k_{\mathrm{B}}$, $k_{\operatorname{T}}$ and $i_{\mathrm{BIC}}$ to be $6$, cf. \eqref{eq:banding_mask} and \eqref{eq:BIC_choice}. Additionally, in case of the \ac{CM} estimators, the measured time as well as the complexity refer to yielding a \ac{CM} estimate and for our proposed estimators they refer to returning an \ac{ICM} estimate. The simulations have been conducted on an \textit{Intel Core(TM) i7-1260P (12th Gen)} processor with a base clock speed of 2.1 GHz and all estimators are implemented in \textit{Matlab} \textit{R2022b}.
%\begin{figure}
%    \centering
%  \include{plots/legend3}
%  \vspace{-32pt}
%\end{figure}
\begin{figure}
\vspace{1.7cm}
\centering
%\resizebox{\columnwidth}{!}{
\begin{subfigure}{.5\columnwidth}
 \includegraphics{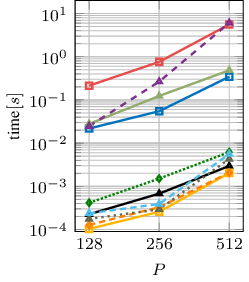}
  \label{fig:sub1}
\end{subfigure}%
\begin{subfigure}{.45\columnwidth}
	  \resizebox{\columnwidth}{!}{
      \begin{tabular}{|c|c|}
        \hline
        \footnotesize Name & \makecell{ \footnotesize Complexity w.o.\\ [-3pt] \footnotesize hyperparameter\\ [-3pt] \footnotesize tuning} \\
        \hline
        \centering \footnotesize \ac{Frob} & \footnotesize $\mathcal{O}(\xi P^2)$ \\[-2.1pt]
        \centering \footnotesize \ac{PGD} &  \footnotesize $\mathcal{O}(L_{\mathrm{B}}T_{\mathrm{B}}P^2)$ \\[-2.1pt]
        \centering \footnotesize \ac{PLS} &  \footnotesize $\mathcal{O}(Pi_{\mathrm{BIC}} + i_{\mathrm{BIC}}^3)$ \\[-2.1pt]
        \centering \footnotesize \ac{band} & \footnotesize $\mathcal{O}(P k_{\mathrm{B}})$  \\[-2.1pt]
        \centering \footnotesize \ac{tape} & \footnotesize $\mathcal{O}(P k_{\mathrm{T}})$  \\[-2.1pt]
        \centering \footnotesize \ac{circ} &  \footnotesize $\mathcal{O}(P^2 \log (P^2))$ \\[-2.1pt]
        \centering \footnotesize \ac{em} & \footnotesize $\mathcal{O}(T_{\mathrm{E}}P^3)$ \\[-2.1pt]
        \centering \footnotesize \ac{avg} & \footnotesize $\mathcal{O}(P^2)$\\[-2.1pt]
        \centering \footnotesize \ac{shu} & \footnotesize $\mathcal{O}(P^3)$ \\[-2.1pt]
        \centering \footnotesize \ac{shb} & \footnotesize $\mathcal{O}(P^2)$ \\
        \hline
    \end{tabular}}
  \label{fig:sub2}
\end{subfigure}
\vspace{-4pt}
\caption{Left: Measured time without hyperparameter tuning for an \ac{ARMA}(1,1) process ($\sigma^2 = 0.64$, $a = 0.7$, $b = 0.3$) with a fixed number of samples $N=64$. Right: Computational complexity without hyperparameter tuning.}
\label{fig:complexity}
\vspace{-0.7cm}
\end{figure}
The computational cost for \ac{band} and \ac{tape} is linear in $P$ and they are the fastest overall. \ac{PLS} is also linear in $P$ and can be implemented in $\mathcal{O}(Pi_{\mathrm{BIC}} + i_{\mathrm{BIC}}^3)$ \acp{FLOP} by considering all of its structure. Thus, its complexity is comparable to \ac{band}'s and \ac{tape}'s complexity. Estimating the convex coefficient for \ac{shb} can be done in $\mathcal{O}(P^2)$ \acp{FLOP}, and its measured time is in \ac{band}'s, \ac{tape}'s and \ac{PLS}'s range, followed by \ac{avg}, also requiring $\mathcal{O}(P^2)$ \acp{FLOP}. Although being a likelihood estimator, \ac{circ} exhibits a closed-form solution and, thus, does not require an iterative procedure. Moreover, it can be implemented by means of the 2D \ac{FFT} in $\mathcal{O}(P^2 \log (P^2))$ \acp{FLOP}. Therefore, it is just slightly slower than \ac{avg}. \ac{PGD} is the fastest of all iterative estimators since its complexity scales quadratically in the samples's dimension $P$, cf. Section \ref{sec:computational_complexity}. \ac{em} exhibits a cubic complexity in the computation of its $T_{\mathrm{E}}$ iterations. In our simulations, however, \ac{em} generally required less iterations than \ac{PGD}, reducing the gap between these two estimators although having different complexity orders of their individual iterations. In our simulations, we used the interior point method of \textit{Matlab}'s built-in optimization function \textit{fmincon} for computing \ac{Frob}. Moreover, the closed-form solution's complexity of the constraint's gradient can be shown to be cubic in $P$ and, thus, has to be approximated to preserve quadratic complexity. We recognized severe speed ups in terms of measured time, but no decrease in performance due to this approximation. Thus,  the update steps of the interior point method can be done in $\mathcal{O}(P^2)$ \acp{FLOP}. However, the scaling factor $\xi$ of the interior point method depends on its specific implementation.

\section{Conclusion}
In this work, we presented a novel class of methods yielding likelihood \ac{CM} and \ac{ICM} estimators for Toeplitz structured \acp{CM} by utilizing the \ac{GS} decomposition. We derived constraint sets for preserving positive definiteness, leading to well-defined likelihood \acp{OP}, and introduced various methods for tuning the hyperparameters and incorporating prior knowledge for this class of estimators. The proposed class of estimators shares the advantage of likelihood estimators, i.e., preserving positive definiteness and considering the underlying distribution, but at the same time includes the advantage of banding and tapering estimators, i.e., an interpretable regularizing method for hyperparameter tuning. Extensive simulations show that the proposed class of estimators significantly outperforms other methods in terms of estimating the \ac{ICM}. In terms of estimating the \ac{CM}, they either show comparable or even better performance to the best benchmark estimators, depending on the underlying structure. %they either perform among the best or are also superior to other methods, depending on the underlying true structure. 
Moreover, a complexity analysis renders our proposed projected LS estimators to be among the computationally cheapest estimators. 
\appendix
%\section{Appendix}
\subsection{Proof of Lemma \ref{lem:frob}}
\label{proof_lemma_frob}
We begin with inserting the \ac{GS} decomposition \eqref{eq:GS_formula} into the definition of positive definiteness, i.e.,
\begin{equation}
    \label{eq:lemma_frob_1}
    \frac{1}{\alpha_0}\left(\bm{x}^{\operatorname{H}}\bm{B}\bm{B}^{\operatorname{H}}\bm{x} - \bm{x}^{\operatorname{H}}\bm{Z}\bm{Z}^{\operatorname{H}}\bm{x}\right) > 0\ \mathrm{for}\ \mathrm{all}\ \bm{x} \in \mathbb{C}^P, \bm{\alpha} \in \mathcal{A}
\end{equation}
with $\bm{x} \neq \bm{0}$. Note that $\bm{B}$ has exactly one eigenvalue $\alpha_0 > 0$ with algebraic multiplicity $P$. Thus, it has full rank and its range is $\mathbb{C}^P$.  By defining $\bm{y} = \bm{B}^{\operatorname{H}}\bm{x}$, we reformulate \eqref{eq:lemma_frob_1} as
\begin{equation}
    \label{eq:lemma_frob_2}
    \| \bm{y} \|_2^2 - \|\bm{Z}^{\operatorname{H}}\bm{B}^{-\operatorname{H}}\bm{y}\|_2^2 > 0\ \mathrm{for}\ \mathrm{all}\ \bm{y} \in \mathbb{C}^P, \bm{\alpha} \in \mathcal{A}
\end{equation}
with $\bm{y} \neq \bm{0}$. By considering the decomposition $\bm{y} = \|\bm{y}\|_2\hat{\bm{y}}$, multiplying out $\|\bm{y}\|_2^2$, and rearranging the terms, we get
\begin{equation}
    \label{eq:lemma_frob_3}
    \|\bm{Z}^{\operatorname{H}}\bm{B}^{-\operatorname{H}}\hat{\bm{y}}\|_2^2 < 1\ \mathrm{for}\ \mathrm{all}\ \hat{\bm{y}} \in \mathbb{C}^P\ \mathrm{with}\ \|\hat{\bm{y}}\|_2^2 = 1, \bm{\alpha} \in \mathcal{A}
\end{equation}
which is equivalent to
\begin{equation}
    \label{eq:lemma_frob_4} 
    \max_{\bm{\alpha} \in \mathcal{A}, \hat{\bm{y}} \in \mathbb{C}^P:\|\hat{\bm{y}}\|_2^2 = 1} \|\bm{Z}^{\operatorname{H}}\bm{B}^{-\operatorname{H}}\hat{\bm{y}}\|_2^2 < 1.
\end{equation}
The maximization over all $\hat{\bm{y}}$ with norm 1 coincides with the definition of the spectral norm, leading to
\begin{equation}
    \label{eq:lemma_frob_5}
    \max_{\bm{\alpha} \in \mathcal{A}} \|\bm{Z}^{\operatorname{H}}\bm{B}^{-\operatorname{H}}\|_2^2 < 1.
\end{equation}
\subsection{Proof of Lemma \ref{lem:g}}
\label{proof_lemma_g}
We start by deriving a closed-form solution for $\bm{B}^{-\operatorname{H}}$. The work in \cite{Sahin2018} shows that the inverse of a $P \times P$ lower triangular Toeplitz matrix $\bm{R}$ of the following form
\begin{equation}
    \bm{R}_{ij} = \begin{cases} 0 & \mathrm{if}\ j > i \\
                                1 & \mathrm{if}\ i = j \\
                                -r_{|i-j|} & \mathrm{else}
                    \end{cases}
\end{equation}
with parameters $\bm{r} = [r_1,\ldots,r_{P-1}]$, is given by

\begin{equation}
\label{eq:fibonacci_inverse}
    (\bm{R}^{-1})_{ij} = \begin{cases}0 & \mathrm{if} j > i \\
                                    1 & \mathrm{if}\ i = j \\
                                    F_{i-j}(\bm{r})& \mathrm{else}
                        \end{cases}
\end{equation}
and $F_{i}(\bm{r})$ is the special case of the generalized Fibonacci sequence defined in \eqref{eq:fibonacci_definition}. By multiplying out $\alpha_0$ in $\bm{B}$, setting $r_i = - \frac{\overline{\alpha_i}}{\alpha_0}$ and swapping the indices due to $\bm{B}^{-\operatorname{H}}$ being upper triangular, we can apply the closed-form solution in \eqref{eq:fibonacci_inverse} to compute
\begin{equation}
    \label{eq:B_inv_closed_form}
    (\bm{B}^{-\operatorname{H}})_{ij} = \begin{cases}0 & \mathrm{if}\ i > j \\
                                    \frac{1}{\alpha_0} & \mathrm{if}\ i = j \\
                                    \frac{1}{\alpha_0} F_{j-i}(- \frac{\overline{\bm{\alpha}_{\geq 1}}}{\alpha_0})& \mathrm{else}
                        \end{cases}
\end{equation}
with $\bm{\alpha}_{\geq 1} = [\alpha_1,\ldots,\alpha_P]^{\operatorname{T}}$. The matrix $\bm{Z}^{\operatorname{H}}$ defined by \eqref{eq:Z_def} can be given elementwise by
\begin{equation}
    \label{eq:Z_elementwise}
    (\bm{Z}^{\operatorname{H}})_{ij} = \begin{cases}\alpha_{P - (j-i)} & \mathrm{if}\ j > i \\ 0 & \mathrm{else} \end{cases}.
\end{equation}
Additionally, the matrix $\bm{Z}^{\operatorname{H}}\bm{B}^{-\operatorname{H}}$ exploits the same structure as $\bm{Z}^{\operatorname{H}}$ since the multiplication of two upper triangular Toeplitz matrices preserves the structure. Therefore, we define
\begin{equation}
    \label{eq:ZB_elementwise}
    (\bm{Z}^{\operatorname{H}}\bm{B}^{-\operatorname{H}})_{ij} = \begin{cases}g_{j-i} & \mathrm{if}\ j > i \\ 0 & \mathrm{else} \end{cases},
\end{equation}where $\bm{g} = [g_1,\ldots,g_{P-1}]$ is fully defined by the first row of $\bm{Z}^{\operatorname{H}}\bm{B}^{-\operatorname{H}}$ and can be elementwise computed by
\begin{equation}
    \label{eq:g_k}
    g_k = \sum_{j=1}^{P-1} (\bm{Z}^{\operatorname{H}})_{0j}(\bm{B}^{-\operatorname{H}})_{jk} = \sum_{j=1}^{k}\frac{\alpha_{P-j}}{\alpha_0}F_{k-j}\left(- \frac{\overline{\bm{\alpha}_{\geq 1}}}{\alpha_0}\right)
\end{equation}
where we used \eqref{eq:Z_elementwise} and the triangular Toeplitz structure of $\bm{B}^{-\operatorname{H}}$. Based on \eqref{eq:ZB_elementwise} we conclude that $g_d$ appears $P-d$ times in $\bm{Z}^{\operatorname{H}}\bm{B}^{-\operatorname{H}}$, rendering the Frobenius norm to be
\begin{equation}
    \label{eq:frob_closed_form}
    \|\bm{Z}^{\operatorname{H}}\bm{B}^{-\operatorname{H}}\|_F^2 = \sum_{d=1}^{P-1} (P-d)|g_d|^2.
\end{equation}
\subsection{Proof of Theorem \ref{th:theorem_K}}
\label{proof_theorem_K}
We bound $g_k$ in \eqref{eq:g_k} using the triangular inequality, i.e.,
\begin{equation}
    |g_k| \leq \sum_{j=1}^{k} \left|\frac{\alpha_{P-j}}{\alpha_0}F_{k-j}\left(- \frac{\overline{\bm{\alpha}_{\geq 1}}}{\alpha_0}\right)\right|.
    \end{equation}
In the following, we utilize the observation that
\begin{align}
    \left|F_i\left(- \frac{\overline{\bm{\alpha}_{\geq 1}}}{\alpha_0}\right)\right| \leq \sum_{l=0}^{i-1}\left|\frac{\overline{\alpha}_{i-l}}{\alpha_0}\right|\left|F_l\left(- \frac{\overline{\bm{\alpha}_{\geq 1}}}{\alpha_0}\right)\right| \leq F_i\left(\left|\frac{\overline{\bm{\alpha}_{\geq 1}}}{\alpha_0}\right|\right),
\end{align}
where the first inequality comes from the triangular inequality and the definition in \eqref{eq:fibonacci_definition}, and the second inequality results from recursively applying the first inequality. Therefore,
\begin{align}
    \left|\frac{\alpha_{P-j}}{\alpha_0}F_{d-j}\left(- \frac{\overline{\bm{\alpha}_{\geq 1}}}{\alpha_0}\right)\right|
    \leq \left|\frac{\alpha_{P-j}}{\alpha_0}\right|F_{d-j}\left(\left|\frac{\overline{\bm{\alpha}_{\geq 1}}}{\alpha_0}\right|\right).
\end{align}
We asumme to have box constraints $|\alpha_i| \leq K_i \alpha_0$ for all $i$ with $\alpha_0 > 0$. Moreover, by observing that $F_{i}(\cdot)$ is componentwise monotonically increasing over the domain $\mathbb{R}^{P-1}_{+}$, we conclude
\begin{equation}
    \left|\frac{\alpha_{P-j}}{\alpha_0}F_{d-j}\left(- \frac{\overline{\bm{\alpha}_{\geq 1}}}{\alpha_0}\right)\right| \leq K_{P-j} F_{d-j}(\bm{K})
\end{equation}
and, according to Lemma \ref{lem:g},
\begin{align}
    \label{eq:frob_smaller_B}
    \|\bm{Z}^{\operatorname{H}}\bm{B}^{-\operatorname{H}}\|_F^2 \leq B(\bm{K})
\end{align}
with $B(\bm{K})$ being defined in \eqref{eq:B_of_K}. Since all $K_i$ in \eqref{eq:B_of_K} are non negative, $B(\bm{K})$ is componentwise monotonically increasing in $\bm{K}$. Additionally, $B(\bm{K})$ is continuous and equals zero for $\bm{K} = \bm{0}$.
Thus,
\begin{equation}
    \exists \bm{K} \in \mathbb{R}^{P-1}_{++}: B(\bm{K}) < 1.
\end{equation}
Based on \eqref{eq:frob_smaller_B} and Lemma \ref{lem:frob}, $B(\bm{K}) < 1$ implies the \ac{GS} decomposition $\bm{\Gamma}(\bm{\alpha})$ in \eqref{eq:GS_formula} to be \ac{PD}, concluding the proof.
\subsection{Proof of Lemma \ref{lem:lemma_f}}
\label{proof_lemma_f}
Since $B(\bm{K})$ is monotonically increasing in each entry of $\bm{K}$ and $f(\eta,i)$ is monotonically increasing in $\eta$ (Property b)), the concatenation $B_f(\eta)$ is monotonically increasing in $\eta$ as well. Additionally, $f(\eta,i)$ is continuous and equals zero for $\eta = 0$ (Property c)). Therefore, $B_f(\eta)$ is continuous, $B_f(0) = 0$, and
\begin{equation}
    \exists \eta \in \mathbb{R}_{++}: B_f(\eta) < 1,
\end{equation}
which implies $\bm{\Gamma}(\bm{\alpha})$ in \eqref{eq:GS_formula} to be \ac{PD} and concludes the proof.

\subsection{Proof of Lemma \ref{lem:lemma_con}}
\label{proof_lemma_con}

It can be seen that the maximization of $\mathcal{L}_{L|C}(\bm{a},\sigma^2)$ defined in \eqref{eq:conditional_log_likelihood} with respect to $\bm{a}$ is $\sigma^2$-independent and, thus, equivalent to maximizing
\begin{equation}
    \label{eq:a_loglikeli}
    - \frac{1}{N}\sum_{n=1}^N\sum_{t=w}^{P-1}|x_t^{(n)} -  \sum_{m=1}^w a_m x_{t-m}^{(n)}|^2.
\end{equation}
Utilizing Wirtinger calculus and setting the derivative of \eqref{eq:a_loglikeli} with respect to $\overline{a}_j$ ($j=1\ldots,w$) to zero yields
\begin{equation}
    \label{eq:mid_identity_conlikeli}
    \frac{1}{N} \sum_{n=1}^N \sum_{t=w}^{P-1}\left(x_t^{(n)} - \sum_{m=1}^w \hat{a}_m x_{t-m}^{(n)}\right) \overline{x}^{(n)}_{t-j} = 0
\end{equation}
for all $j=1,\ldots,w$. Rearranging the terms in \eqref{eq:mid_identity_conlikeli} leads to
\begin{equation}
    \label{eq:later_identity_conlikeli}
    \sum_{m=1}^w \hat{a}_m \sum_{t=w}^{P-1} \frac{1}{N} \sum_{n=1}^N x_{t-m}^{(n)} \overline{x}_{t-j}^{(n)} = \sum_{t=w}^{P-1} \frac{1}{N} \sum_{n=1}^N x_t^{(n)} \overline{x}_{t-j}^{(n)}
\end{equation}
for all $j=1,\ldots,w$. By inserting the definition of $\tilde{\bm{S}}$ in \eqref{eq:S_tilde} and rewriting \eqref{eq:later_identity_conlikeli} in matrix-vector notation, we end up with
\begin{equation}
    \tilde{\bm{S}}_{\geq 1, \geq 1} \hat{\bm{a}} = \tilde{\bm{S}}_{\geq 1, 0}.
\end{equation}
For deriving the closed form of $\hat{\sigma}^2$ in \eqref{eq:closed_form_AR}, we first observe that by multiplying \eqref{eq:mid_identity_conlikeli} with $\hat{a}_j$ and summing over all $j = 1,\ldots,w$, we yield
\begin{equation}
    \label{eq:sigma_identity}
    \frac{1}{N}\sum_{n=1}^N \sum_{t=w}^{P-1}\left( \left|\sum_{j=1}^w \hat{a}_j x_{t-j}^{(n)}\right|^2 - x_t^{(n)} \sum_{j=1}^w \overline{a}_j \overline{x}_{t-j}^{(n)} \right) = 0.
\end{equation}
Setting the derivative of $\mathcal{L}_{L|C}(\hat{\bm{a}},\sigma^2)$ with respect to $\sigma^2$ to zero and rearranging the terms leads to
\begin{align}
    \hat{\sigma}^2 = & \frac{1}{N(P-w)} \sum_{n=1}^N \sum_{t=w}^{P-1} \left|x_t^{(n)} - \sum_{m=1}^w \hat{a}_m x_{t-m}^{(n)}\right|^2 \\\nonumber
    = & \frac{1}{N(P-w)} \sum_{n=1}^N \sum_{t=w}^{P-1} |x_t^{(n)}|^2 - \overline{x}_t^{(n)}\sum_{m=1}^w \hat{a}_m x_{t-m}^{(n)} \\ & - x_t^{(n)} \sum_{m=1}^w \overline{\hat{a}}_m\overline{x}_{t-m}^{(n)} + \left|\sum_{m=1}^{w} \hat{a}_m x_{t-m}^{(n)}\right|^2 \\ 
    = & \frac{1}{N(P-w)} \sum_{n=1}^N \sum_{t=w}^{P-1} |x_t^{(n)}|^2 - \overline{x}_t^{(n)}\sum_{m=1}^w \hat{a}_m x_{t-m}^{(n)},
\end{align}
where the third equality follows from \eqref{eq:sigma_identity}. By inserting the definition of $\tilde{\bm{S}}$ from \eqref{eq:S_tilde}, we end up with
\begin{equation}
    \hat{\sigma}^2 = \frac{1}{P-w}\left( \tilde{\bm{S}}_{0,0} - \tilde{\bm{S}}_{\geq 1,0}^{\operatorname{H}}\hat{\bm{a}} \right)
\end{equation}
which concludes the proof. 
\subsection{Proof of Lemma \ref{lem:lemma_tr}}
\label{proof_lemma_tr}
Let $\tilde{\bm{D}} = \bm{D}(\operatorname{\mathbf{E}}^k)^{\operatorname{T}}$.
It can be seen that
\begin{equation}
    \label{eq:D_tilde_element}
    \tilde{D}_{ij} = \begin{cases}0 & \mathrm{if}\ (j<k) \lor (i < j-k)\\ d_{|i-j+k|}\ & \mathrm{else}\end{cases}.
\end{equation}
Therefore, we can write
\begin{equation}
    \tr(\bm{C}_{\mathrm{T}}\bm{D}(\operatorname{\mathbf{E}}^k)^{\operatorname{T}}) = \sum_{i,j = 0}^{P-1}c_{|i-j|}\tilde{D}_{ji} = \sum_{j=k}^{P-1}\sum_{i = j-k}^{P-1}c_{|i-j|}d_{|i-j+k|}.
\end{equation}
By re-indexing $m = i-j+k$, we get
\begin{equation}
    \label{eq:lemma_tr_zwi}
    \tr(\bm{C}_{\mathrm{T}}\bm{D}(\operatorname{\mathbf{E}}^k)^{\operatorname{T}}) = \sum_{j=k}^{P-1}\sum_{m=0}^{P-1+k-j} c_{|m-k|}d_m.
\end{equation}
The double-sum's summation region in the space spanned by the indices $j$ and $m$ corresponds geometrically to a rectangular and a triangular. From this consideration, it can be seen that $\tr(\bm{C}_{\mathrm{T}}\bm{D}(\operatorname{\mathbf{E}}^k)^{\operatorname{T}})$ equals
\begin{equation}
\begin{aligned}
  \sum_{m=0}^{k-1}\sum_{j=k}^{P-1}c_{|m-k|}d_m + \sum_{m=k}^{P-1}\sum_{j=k}^{P-1+k-m}c_{|m-k|}d_m \\=(P-k) \sum_{m=0}^{k-1}c_{|m-k|}d_m + \sum_{m=k}^{P-1}(P-m)c_{|m-k|}d_m.
\end{aligned}
\end{equation}
\subsection{Proof of Lemma \ref{lem:lemma_tr2}}
\label{proof_lemma_tr2}
Similarly to the proof in Appendix \ref{proof_lemma_tr}, we start defining $\tilde{\bm{D}} = \bm{D}(\operatorname{\mathbf{E}}^k)^{\operatorname{T}}$, which is given in \eqref{eq:D_tilde_element}. By following Appendix \ref{proof_lemma_tr} and re-indexing $m=k+j$ and $b=j-k$, we get
\begin{equation}
    \tr(\bm{Q}\bm{D}(\operatorname{\mathbf{E}}^k)^{\operatorname{T}}) = \sum_{b=0}^{P-1-k}\sum_{m=0}^{P-1-b} \bm{Q}_{b+k,b+m}d_m.
\end{equation}
Similarly to the proof in Appendix \ref{proof_lemma_tr}, the double-sum's summation region corresponds to a triangular and a rectangular in the space spanned by he indices $m$ amd $b$. From this consideration, it can be seen that $\tr(\bm{Q}\bm{D}(\operatorname{\mathbf{E}}^k)^{\operatorname{T}})$ equals
\begin{equation}
    \sum_{m=0}^{k-1} d_m \sum_{b=0}^{P-1-k} \bm{Q}_{b+k,b+m} + \sum_{m=k}^{P-1} d_m \sum_{b=0}^{P-1-m} \bm{Q}_{b+k,b+m},
\end{equation}
which concludes the proof.
\bibliographystyle{IEEEtran.bst}
\bibliography{references.bib}
\newpage
\end{document}